\documentclass[10pt, twocolumn]{IEEEtran}

\usepackage{amsmath}
\usepackage{amssymb}    % mathbb
\usepackage{amsfonts}
\usepackage[english]{babel}
\usepackage{cite} %reference contraction
\usepackage{multirow}
\usepackage{stfloats}    % big equations float

\usepackage{algorithm}   %format of the algorithm
\usepackage{algorithmic} %format of the algorithm

\usepackage{caption} % Figure to Fig.
\usepackage{graphicx}
\usepackage{epsfig}
\usepackage{subfigure} %subfigs
\usepackage{tablefootnote}

\usepackage{accents}
\makeatletter
\def\widebar{\accentset{{\cc@style\underline{\mskip10mu}}}}
\def\Widebar{\accentset{{\cc@style\underline{\mskip13mu}}}}
\makeatother

\newtheorem{theorem}{Theorem}

\newtheorem{definition}{Definition}
\newtheorem{proposition}{Proposition}

\newtheorem{example}{Example}

\usepackage{framed}   % color for revisions
\usepackage{color,soul}
\definecolor{shadecolor}{rgb}{0.67, 0.9, 0.93}%blizzardblue

\definecolor{aliceblue}{rgb}{0.94, 0.97, 1.0}
\definecolor{blizzardblue}{rgb}{0.67, 0.9, 0.93}
\definecolor{antiquebrass}{rgb}{0.8, 0.58, 0.46}
\definecolor{beaublue}{rgb}{0.74, 0.83, 0.9}
\definecolor{deepskyblue}{rgb}{0.0, 0.75, 1.0}
\sethlcolor{blizzardblue}

\begin{document}

\captionsetup[figure]{labelfont={ }, name={Fig.}, labelsep=period} % Figure to Fig.
\pagestyle{empty}

\title{Optimal Power Control for Transmitting Correlated Sources with  Energy Harvesting Constraints }

\author{Yunquan~Dong,~\IEEEmembership{Member,~IEEE},~ Zhi~Chen,~\IEEEmembership{Member,~IEEE},~\\
        Jian Wang,~\IEEEmembership{Student Member,~IEEE},~and~Byonghyo~Shim,~\IEEEmembership{Senior Member,~IEEE}

        \thanks{ Y. Dong is with the School of Electronic and Information Engineering,  Nanjing University of Information Science and Technology, Nanjing 210044, China. (e-mail: yunquandong@nuist.edu.cn).

        Z. Chen is with the Department of Electronic and Computer Engineering, University of Waterloo, Waterloo, ON N2L3G1, Canada (e-mail: chenzhi2223@gmail.com).

                        J. Wang and B. Shim are with the Institute of New Media and Communications and the Department  of Electrical and Computer Engineering, Seoul National University, Seoul 151744, Korea (e-mail: jwang@islab.snu.ac.kr; bshim@snu.ac.kr).

        This work is supported by the National Research Foundation of Korea (NRF) grant funded by the Korea government (MSIP 2014R1A51011478).
        }
        }

\date{}
\maketitle
\thispagestyle{empty}

\vspace{-15mm}
\begin{abstract}
We investigate the weighted-sum distortion minimization problem in transmitting two correlated Gaussian sources over Gaussian channels using two energy harvesting nodes.
    To this end, we develop offline and online power control policies to optimize the transmit power of the two nodes.
In the offline case, we cast the problem as a convex optimization and investigate the structure of the optimal solution.
    We also develop a  generalized water-filling based power allocation algorithm to obtain the optimal solution efficiently.
For the online case, we quantify the distortion of the system using  a cost function  and show that the expected cost equals the expected weighted-sum distortion.
    Based on Banach's fixed point theorem, we further propose a geometrically converging algorithm to find the minimum cost via simple iterations.
Simulation results show that our online power control  outperforms the greedy power control where each node uses all the available energy in each slot and  performs close to that of the proposed offline power control.
Moreover, the performance of our offline power control almost coincides with the performance limit of the system.
%Moreover, while the performance of our online power control performs close to that of the proposed offline power control, the performance of our offline power control almost coincides with the performance limit of the system.
\end{abstract}


\section{Introduction}
In energy harvesting networks, each node  continually acquires energy from nature or man-made phenomenon \cite{Ulukus-2015-review}.
     This feature has made energy harvesting a key technology to prolong the life-time of wireless sensor networks (WSNs), where sensors are often deployed in some unreachable areas~\cite{WSN-2011-survey}.
For energy harvesting powered WSNs, however, formidable challenges still remain since the energy arrivals of each node are often sporadic and irregular.
    To address this issue, many energy scheduling schemes optimizing the information transmission of energy harvesting communication systems have been suggested in recent years~\cite{Ulukus-2015-review, WSN-2011-survey, Ulukus-2011-policy, Ulukus-2012-packet_mac, RuiZhang-2012-tsp, Ydong-2015-JSAC, Zeng-2015-TWC, Badiei-2014-TIT, Sakulkar-2016-Arxiv, Dong-2016-ICCC, Deniz-2015, Knorn-2015-TSP, Nourian-2015-JSAC, Bhat-2016-ICC, Ozcelik-2016-ISIT, NIT-2011, Yzhao-2015-TWC, ZChen-2016-JSAC,Cui-2014-TWC, Dongin-TWC-2016}.

First, if the harvesting process is fully predictable (i.e., known non-causally at transmitter), the harvested energy can be scheduled in an \textit{offline} manner~\cite{Ulukus-2011-policy, Ulukus-2012-packet_mac, RuiZhang-2012-tsp}.
    In this scenario,  the energy scheduling for the  transmission process turns to be deterministic and thus can be solved before the transmission actually happens.
Second, if the energy harvesting process cannot be well predicted, \textit{online} energy scheduling is required, in which each node adjusts its transmit power  based on previous and current energy states in real-time~\cite{Ydong-2015-JSAC, Zeng-2015-TWC, Badiei-2014-TIT,Sakulkar-2016-Arxiv, YDong-2016-JSAC, Dongin-TWC-2016}.
    This online energy scheduling has been often modeled as Markov Decision Processes (MDP) and solved by Dynamic Programming (DP) \cite{Sakulkar-2016-Arxiv, Badiei-2014-TIT, Zeng-2015-TWC}.
There have also been some works considering both offline and online policies, e.g.,~\cite{RuiZhang-2012-tsp, ZChen-2016-JSAC, Cui-2014-TWC}.
    As is expected in these works, offline energy scheduling policies always outperform their online counterparts owing to the non-causal information on the energy harvesting process at transmitters.
{In a nutshell,  the offline power control scheme requires non-causal information about the energy harvesting process and thus outperforms the online power control scheme but is less practical.
    Whereas, online power control schemes are more practical and may approach the performance of offline schemes, but solving the optimal schemes using the MDP model is generally difficult and even intractable in some scenarios.}
% {In summary, while offline power control schemes require non-causal information about the energy harvesting process and hence is less practical, they perform better;
%    although the online power control schemes are more implementable and may approach the performance of offline schemes, solving the optimal schemes using MDP model is generally difficult and even intractable in some scenarios.}

In WSNs,  collected information (e.g., temperature, humidity, pollution density)  is usually continuous and can be compressed before being transmitted to the fusion center.
    Since the compression process  inevitably introduces some distortion to these information, it is of importance  to schedule the harvested energy so that the distortion caused by the recovering process can be minimized.
To this end, both offline and online power control policies (scheduling energy by controlling transmit powers of nodes) have been widely used~\cite{Nourian-2015-JSAC, Knorn-2015-TSP, Bhat-2016-ICC, Ozcelik-2016-ISIT,  Yzhao-2015-TWC}.
    { In fact, both the reliability and the efficiency of extracting  information from the recovered samples are dominated by the distortion of the network.}
    In these works, the fusion center tries to recover the uncoded signals using mean-squared error (MSE) estimators~\cite{Ozcelik-2016-ISIT, Yzhao-2015-TWC} or  best linear unbiased estimators (BLUE)~\cite{Nourian-2015-JSAC, Knorn-2015-TSP}.
 {Moreover, since the environmental information collected by adjacent nodes is highly correlated with each other, the energy efficiency of the network can be increased by removing the redundancy among these samples using distributed lossy source coding, i.e., the rate-distortion theory for multi-source networks}~\cite{NIT-2011}.
     {However, the task to characterize the rate-distortion region and formulate the corresponding distortion minimization problem is very difficult.}
In fact, previous studies focused only on the problem with a tractable static setting where the channels  between the two nodes and the fusion center are symmetric \cite{Deniz-2015} or using the offline power control for the non-static case~\cite{Dong-2016-ICCC}.
     {Hence, achieving the information theoretic performance limit of more general networks using more practical online power control remains an  open problem.}

 {In this paper, we study both offline and online power control policies minimizing the information theoretic distortion of the system, where two correlated sources are transmitted over non-symmetric Gaussian channels using energy harvesting nodes.
    We first minimize the weighted-sum distortion over a finite period via optimal offline power control.
We then consider the online case and investigate the optimal power control minimizing the expected weighted-sum distortion.
    In particular, we propose a cost function to quantify the distortion of the system, which is proved to be equal to the expected weighted-sum distortion.
Based on Banach's fixed point theorem, we further present an algorithm approaching the minimum expected distortion via simple iterations. }
    The main contributions of the paper can be summarized as follows:
\begin{itemize}
  \item We present the structure of optimal offline power allocations.
                We show that the energy buffer should be depleted if the averaged harvested energy  in future slots is larger than that of previous slots.
            Moreover, the transmit power of a node should be increased after the slots in which its energy buffer is depleted and should be decreased after the slots in which the energy buffer of the other node is depleted.
   \item We propose an iterative algorithm to solve the optimal offline power allocation.
                By optimizing the transmit power of each node separately and running the single user optimization iteratively, the algorithm converges to the optimal solution in a small number of iterations.
  \item We propose a cost function for online power control.
                Since it is hard to analyze the expected weighted-sum distortion directly, we quantify the system distortion using a cost function, which is defined as the weighted sum of current distortion and expected future distortion.\footnote{It should be noted that the proposed cost function model is different from the discounted cost model or the average cost model in traditional MDP theory \cite{Pturmsn-2014-MDP}.}
            We further prove that the expected cost equals the expected weighted-sum distortion.
  \item We prove that the minimum expected cost is the fixed point of some mapping.
                We then present an algorithm approaching the optimal cost using simple iterations.
\end{itemize}

This paper is organized as follows.
   Section~\ref{sec:2_model} presents the network model and the distortion minimization problem.
In Section~\ref{sec:3offline}, we present the structure of the optimal offline power allocation and propose an algorithm to obtain  the solution efficiently.
    In Section~\ref{sec:4online}, we consider the online power allocation and propose an algorithm to solve the problem using simple iterations.
Finally, the numerical results are provided in Section~\ref{sec5:simulation} and  our work is concluded in Section~\ref{sec6:conclusion}.

\textit{Notations:} We use boldface letters to denote vectors and matrices, use $\tau=1,\cdots,T$ to index time, use $k=1,2$ to index nodes, and use $\tilde{k}=3-k$ to refer to the other node.
    $\mathbb{R}_{++}^n$ and $\mathbb{Z}_{++}^n$ denote the $n$-dimensional vector of positive real numbers and positive integers, respectively.
In addition, $(\cdot)^{\text{T}}$ denotes the transpose operation.

\section{System Model}\label{sec:2_model}
\subsection{Network Model}
We consider the system of two sensor nodes and a fusion center, where each sensor node is equipped with an energy harvesting device and a transmit module, as shown in Fig.~\ref{fig:net_model}.
    The sensors observe environmental information (e.g., temperature and humidity) and then send the sampled data to the fusion center using the energy harvested from ambient environments.
We assume that time is slotted and consider a duration of $T$ slots.
     In each slot $\tau\in[1,\cdots, T]$, each node acquires an information sample $X_{k\tau}$ ($k=1,2$).
 It is assumed that $X_{1\tau}$ and $X_{2\tau}$ are zero-mean, unit-variance Gaussian random variables with correlation coefficient $\sqrt{\eta}$  ($0<\eta<1$).
     {We also assume that the slot length is large enough so that the samples of different slots are  independent from each other} \cite{Deniz-2015,Dong-2016-ICCC,Cui-2007-TSP}.

\begin{figure}[!t]
\centering
\includegraphics[width=2.8in]{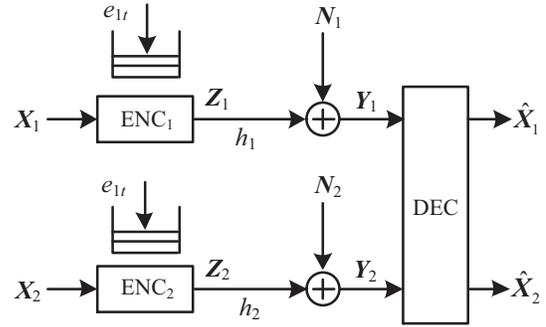}
\caption{Source coding model and information transmission model, where  %$(X_{1t}, X_{2t})\sim \mathcal{N}(0, \sqrt{\eta})$ are zero-mean correlated Gaussian sources,
ENC and DEC are the encoder at each node and the decoder at the fusion center, respectively. } \label{fig:net_model}
\end{figure}

 Before being transmitted to the fusion center through finite-capacity channels,   the correlated samples $X_{1\tau}$ and $X_{2\tau}$ need to be compressed using distributed lossy source coding~\cite{NIT-2011}.
    Let  $r_{1\tau}$ and $r_{2\tau}$ be the  coding rate of the two sources.
    Afterwards, the obtained messages are encoded into channel codewords $\boldsymbol{Z}_{1\tau}$ and $\boldsymbol{Z}_{2\tau}$, respectively.
Upon receiving $\boldsymbol{Y}_{k\tau}$, the fusion center decodes the messages and then restores  the transmitted samples $X_{k\tau}$ with some distortion.

    We assume that the channel between each node and the fusion center is an additive white Gaussian noise (AWGN) channel with static channel  gain $h_k$ and zero-mean, unit-variance Gaussian noise, i.e., $N_{k\tau}\sim \mathcal{N}(0,1)$.
We also assume that the two nodes transmit information using two distinct frequency bands.
    Under this setting, the source-channel separation is optimal~\cite{Luo-2007}.
    Since the slot length is large, we can readily assume that Shannon capacity is achievable. In this case, the maximum transmission rate over each channel is
\begin{eqnarray*}
  c_{k\tau} \hspace{-3mm}&=&\hspace{-3mm} \frac12 \log_2(1+h_k p_{k\tau}),~~~~k=1,2,
\end{eqnarray*}
where $p_{k\tau}$ is the transmit power of node $k$ in the $\tau$-th slot.

The two nodes are both equipped with an energy buffer, where the buffer sizes are denoted as $L_1$ and $L_2$, respectively.
    We assume that the energy buffers are large and the probability of energy overflow is negligible.
 {For example, the capacity of a small button battery is more than 200 milliampere hour (mAh), which is large enough for most energy harvesting scenarios}~\cite{button-2015}.
     {Moreover, a reasonable power control will try to avoid energy overflow to maximize the energy efficiency of the network.
Thus, we do not consider the constraint of finite buffer size in this paper.}
    In each slot $\tau$, node $k$ harvests $e_{k\tau}$ units of energy (normalized by slot length so that we can use energy and power interchangeably) and put the energy into  energy buffer.
 {We assume that the harvested energy in current slot can be used either in current slot or in future slots.}
    On the contrary, since the energy harvested in future slots cannot be used in  the current slot, the transmit power of each user must obey the following energy causality constraint:
\begin{equation}\label{CSTR_energy_causality}
  \sum_{i=1}^\tau p_{ki} \leq \sum_{i=1}^\tau e_{ki}, \qquad k=1,2, ~\tau=1,\cdots,T.
\end{equation}

\subsection{Rate-Distortion Model}
    Under the squared-error measure $d(x, \hat{x})=(x- \hat{x})^2$, the rate-distortion region $\mathcal{R}(D_{1}, D_{2})$ of two zero-mean, unit variance correlated Gaussian sources is the intersection of the following three regions~\cite[Chap.~12, Theorem 3]{NIT-2011}:
\begin{eqnarray}
  \label{eq:region_1}
  \mathcal{R}_1(D_{1}) \hspace{-3mm}&=&\hspace{-3mm} \left\{\hspace{-.6mm} (r_{1}, r_{2}):r_{1}\geq \mathrm{R} \hspace{-.6mm}
                    \left( \frac{\bar{\eta}+\eta2^{-2r_{2}}}{D_{1}}\right) \hspace{-.7mm} \right\} ,\\
  \label{eq:region_2}
  \mathcal{R}_2( D_{2}) \hspace{-3mm}&=&\hspace{-3mm}\left\{ \hspace{-.6mm} (r_{1}, r_{2}):r_{2}\geq \mathrm{R} \hspace{-.6mm}
                    \left( \frac{\bar{\eta}+\eta2^{-2r_{1}}}{D_{2}}\right) \hspace{-.7mm} \right\} ,\\
  \label{eq:region_3}
  \mathcal{R}_{12}(D_{1}, D_{2}) \hspace{-3mm}&=&\hspace{-3mm}\left\{ \hspace{-.6mm} (r_{1}, r_{2}):
                    r_{1}\hspace{-.3mm}+\hspace{-.3mm}r_{2} \hspace{-.1mm} \geq \hspace{-.1mm} \mathrm{R} \hspace{-.65mm}
                    \left( \frac{\bar{\eta}\phi(D_{1},D_{2})}{2D_{1}D_{2}}\right) \hspace{-.7mm} \right\}
\end{eqnarray}
where $\bar{\eta}=1-\eta$, $\phi(D_{1},D_{2})=1+\sqrt{1+4\eta D_{1}D_{2}/\bar{\eta}^2}$ and $\mathrm{R}(x)=\frac12 \log_2x$.
     {This rate-distortion region follows the Berger--Tung inner bound} \cite[Chap.~12, Theorem 1]{NIT-2011}.  {Note that} \eqref{eq:region_1}   {and} \eqref{eq:region_2} {present the rate-distortion trade-off of source 1 and source 2, respectively,  and} \eqref{eq:region_3}  {indicates the joint constraint on the two sources.
In addition, since  two sources are correlated with each other, the achievable distortion of either source is closely related with both $r_1$ and $r_2$}.

Since both the transmit power of  nodes and the capacity of node-receiver channels  are limited, the coding rate $r_{k\tau}$ must satisfy
\begin{equation}\label{CSTR_rate}
   r_{k\tau} \leq c_{k\tau}, \qquad ~~ k=1,2, ~~\tau=1,\cdots,T.
\end{equation}
 For a given coding rate pair $[r_{1\tau},r_{2\tau}]$, by solving $D_{1\tau}$ and $D_{2\tau}$ from the rate-distortion region~\eqref{eq:region_1}--\eqref{eq:region_3}, one can show that the achievable distortion pair satisfies
\begin{eqnarray}
\label{eq:cstr6}
  D_{1\tau} \hspace{-3mm}&\geq&\hspace{-3mm}  (\bar{\eta}+\eta2^{-2r_{2\tau}})2^{-2r_{1\tau}}\triangleq d_{1\tau}^{\min} , \\
  D_{2\tau} \hspace{-3mm}&\geq&\hspace{-3mm}  (\bar{\eta}+\eta2^{-2r_{1\tau}})2^{-2r_{2\tau}} \triangleq d_{2\tau}^{\min} ,\\
  \label{eq:cstr8}
  D_{1\tau}D_{1\tau}\hspace{-3mm}&\geq&\hspace{-3mm} (\bar{\eta}+\eta2^{-2(r_{1\tau}+r_{2\tau})})2^{-2(r_{1\tau}+r_{2\tau})} \triangleq d_{12\tau}^{\min} .
\end{eqnarray}
Since  $r_{1\tau}$ and $r_{2\tau}$ are non-negative, we also have
\begin{eqnarray}
  D_{1\tau} \hspace{-3mm}&\leq&\hspace{-3mm}  \bar{\eta}+\eta2^{-2r_{2\tau}} \triangleq D_{1\tau}^{\max} , \\ \label{CSTR_d_max}
  D_{2\tau} \hspace{-3mm}&\leq&\hspace{-3mm}  \bar{\eta}+\eta2^{-2r_{1\tau}} \triangleq D_{2\tau}^{\max}.
\end{eqnarray}

From constraints \eqref{eq:cstr6}--\eqref{CSTR_d_max}, we know that all achievable distortion pairs must appear in the shaded area of Fig. \ref{fig:d1d2}.

\begin{figure}[!t]
\centering

\includegraphics[width=3.0in]{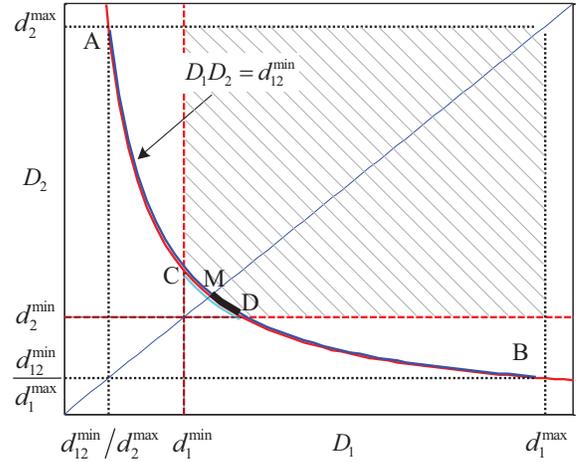}
\caption{Achievable distortion region (the shaded area) for a given coding rate pair $(r_{1\tau},r_{2\tau})$. %The coordinate of the points are: A
}   \label{fig:d1d2}
\end{figure}

Let $w_1$ and $w_2$ be two positive weighting coefficients satisfying $w_1+w_2=1$.
    Without loss of generality, we assume that $w_1<w_2$.
For a given a coding rate $\boldsymbol{r}=[r_1,r_2]$, we define the \textit{minimum weighted-sum distortion} as
\begin{equation}\label{df:w_dis}
    D({\boldsymbol{ r}})= \mathop{\mathrm{minimize}~}\limits_{[D_1,D_2]\in \mathcal{R}(D_1,D_2)} \left(w_1D_1+w_2D_2\right).
\end{equation}

First, the weighting coefficients $w_1$ and $w_2$ implies the priority of how much the distortion of each node contributes to the system performance.
    Also, it can be seen from in Fig. \ref{fig:d1d2} that for any given weighting coefficient pair $[w_1,w_2]$,  the minimum weighted-sum distortion $D({\boldsymbol{ r}})$ occurs at some point on both the line $w_1D_1+w_2D_2=c_0$ and  the distortion region boundary, where $c_0$ is a certain constant.
Thus, by adjusting $[w_1,w_2]$ and solving the corresponding minimum weighted-sum distortion, we can obtain a full characterization of the achievable distortion region, as well as a systematic evaluation of the  validity and the reliability of the recovered samples.
    Since the  rate $\boldsymbol{r}=[r_1,r_2]$ is a function of the transmit power $\boldsymbol{p}=[p_1,p_2]$, we also denote the minimum weighted-sum distortion as $D({\boldsymbol{ p}})$ in some cases, e.g.,  in Section \ref{sec:4online}.

In this paper, we aim at minimizing the weighted-sum distortion by scheduling the harvested energy.
    In particular, the following proposition characterizes the minimum weighted-sum distortion $ D(\boldsymbol{r}) $  explicitly.
\begin{proposition} \label{lem:covex1}
     $D({\boldsymbol{ r}})$ is convex in coding rate $\boldsymbol{r}=[r_1,r_2]$. Moreover, $ D(\boldsymbol{r}) $ is given by
    \begin{equation}\label{rt:D_r}
      D(\boldsymbol{r}) =\left\{
            \begin{aligned}
                    & 2\sqrt{w_1w_2 d_{12}^{\min}}  & \textrm{if}~ r_2\geq g(r_1),\\
                    & w_1 \frac{d_{12}^{\min}}{d_{2}^{\min}} + w_2 d_2^{\min}  & \textrm{if}~r_2< g(r_1),
            \end{aligned}
            \right.
    \end{equation}
    where $g(r)=-\frac12 \log_2 \frac{w_1\bar{\rho}2^{-2r}} {w_2(\bar{\rho}+\rho2^{-2r})^2-w_1\rho 2^{-4r}}$.
\end{proposition}

\begin{proof}
    See Appendix \ref{prf:lem1}.
\end{proof}

Since $D(\boldsymbol{r})$ is convex, the offline distortion minimization problem can be solved using  standard optimization techniques. We will discuss more on this in Section \ref{sec:3offline}.

\subsection{Problem Formulation}
In this paper, we consider both offline and online power allocation schemes to minimize the weighted-sum distortion.
      {In the offline case, we assume that the energy harvesting process is known non-causally at the two nodes. }
Thus, the weighted-sum distortion over a certain period can be minimized and the optimal power control can be obtained before the real transmission, by solving the following optimization problem:
\begin{eqnarray*}
 \mathop{\mathrm{minimize}~}\limits_{p_{k\tau}} \hspace{-3mm}& &  \hspace{-2mm}\sum_{\tau=1}^T \left( w_1 D_{1\tau} +w_2 D_{2\tau} \right) \\
  (\mathrm{P}_1)~~\mathrm{subject~to} \hspace{-3mm}&&\hspace{-2mm}  \text{Eq.}~\eqref{CSTR_energy_causality}, \eqref{CSTR_rate}-\eqref{CSTR_d_max},\\
                        \hspace{-2mm}&&\hspace{-2mm} p_{k\tau} \geq 0, ~ k=1,2,~\tau=1,\cdots,T.
\end{eqnarray*}

 {For the online power control, the sensor nodes are unaware of the energy harvesting process.
    Nevertheless, we assume that the distribution of the energy harvesting process is known to the nodes and thus the nodes can  adjust their transmit power based  on their causal energy status. }

Let $\varrho=\{\rho_1,\cdots,\rho_T\}$ be the power control policy that maps the energy state (remaining energy of nodes) of the system to the transmit power of each node.
    We then minimize the expectation of the weighted-sum distortion by solving the following problem:
\begin{eqnarray*}
 \mathop{\mathrm{minimize~}}\limits_{\varrho} \hspace{-3mm}& &  \hspace{-2mm} \lim_{T\rightarrow\infty}
                                                                            \mathbb{E}\left(\frac{1}{T}\sum_{\tau=1}^T \left(w_1 D_{1\tau} +w_2 D_{2\tau}\right) \right)\\
  (\mathrm{P}_2)~~\mathrm{subject~to} \hspace{-3mm}&&\hspace{-2mm} \text{Eq.}~\eqref{CSTR_energy_causality}, \eqref{CSTR_rate}-\eqref{CSTR_d_max}, \\
                        \hspace{-2mm}&&\hspace{-2mm} p_{k\tau} \geq 0, ~ k=1,2,~\tau=1,\cdots,T.
\end{eqnarray*}

%\begin{figure}[!t]
%\centering
%
%\includegraphics[width=3.2in]{f2_d1d2.eps}
%\caption{The achievable distortion region (the shaded area) for a given coding rate pair $(r_{1\tau},r_{2\tau})$. %The coordinate of the points are: A
%}   \label{fig:d1d2}
%\end{figure}

\section{Offline Power Control} \label{sec:3offline}
When the energy harvesting process is known non-causally, we can solve $(\mathrm{P}_1)$ using KKT conditions \cite{cv_byod-2004}.
    Further,  we show that the optimal offline power control can be explained as a generalized water-filling problem.

\subsection{Standard Formulation}
    To utilize the channels  efficiently, we assume $r_{k\tau}=c_{k\tau}$ in each slot for each node.
Using variable substitution $p_{k\tau}=\frac{1}{h_k}(2^{2r_{k\tau}}-1)$ and after some manipulations, the optimization problem $(\mathrm{P}_1)$ can be expressed in the standard convex optimization form as
\begin{eqnarray}
\nonumber
\hspace{-6mm}  \mathop{\mathrm{minimize~}}\limits_{ r_{k\tau}} ~~~\hspace{-3mm}& &  \hspace{-3mm}\sum_{\tau=1}^T \left( w_1 D_{1\tau} +w_2 D_{2\tau} \right) \\
  \label{cst:lambda1}
\hspace{-6mm} (\mathrm{P}_3)~~\mathrm{subject~to} \quad \hspace{-3mm}&&\hspace{-3mm} -D_{1\tau}+(\bar{\eta}+\eta2^{-2r_{2\tau}})2^{-2r_{1\tau}} \leq 0,\\
                            \label{cst:lambda2}
\hspace{-6mm}                        \quad \hspace{-3mm}&&\hspace{-3mm} -D_{2\tau}+(\bar{\eta}+\eta2^{-2r_{1\tau}})2^{-2r_{2\tau}} \leq 0, \\
 \nonumber
\hspace{-6mm}                        \quad \hspace{-3mm}&&\hspace{-3mm}
                                                                                                                                        -\log_2 D_{1\tau}-\log_2 D_{2\tau} - 2(r_{1\tau}+r_{2\tau}) \\
                            \label{cst:lambda3}
\hspace{-6mm}                        \quad \hspace{-3mm}&&\hspace{-3mm}
                                                                                                                                            + \log_2(\bar{\eta}+\eta2^{-2(r_{1\tau}+r_{2\tau})}) \leq 0,\\
                         \label{cst:lambda4}
\hspace{-6mm}                        \quad \hspace{-3mm}&&\hspace{-3mm} (D_{1\tau}-\bar{\eta})2^{2r_{2\tau}} -\eta \leq 0, \\
                         \label{cst:lambda5}
                        \quad \hspace{-3mm}&&\hspace{-3mm} (D_{2\tau}-\bar{\eta})2^{2r_{1\tau}} -\eta \leq 0, \\
\hspace{-6mm}                        \quad \hspace{-3mm}&&\hspace{-3mm} \sum_{i=1}^\tau \frac{1}{h_k}(2^{2r_{ki}}-1) - \sum_{i=1}^\tau e_{ki} \leq 0;\\
                         \label{cst:theta}
                        \quad \hspace{-3mm}&&\hspace{-3mm} -r_{k\tau} \leq 0,  %\\ \nonumber
%\hspace{-6mm}                        \quad \hspace{-3mm}&&\hspace{-3mm}
                                                                                                                                            \forall~k=1,2,~ \tau=1,\cdots,T.
\end{eqnarray}

    As is shown in Proposition \ref{lem:covex1}, the objective function of $(\mathrm{P}_3)$ is convex.
The achievable region of $(D_{1\tau}, D_{2\tau})$  has been shown to be convex in \cite{Deniz-2015}.
    Moreover, one can show that the constraints \eqref{cst:lambda1}--\eqref{cst:theta} of  $(\mathrm{P}_3)$  are all convex.
Thus, $(\mathrm{P}_3)$ is a convex optimization problem and can be solved using  KKT conditions \cite{cv_byod-2004}.
    The corresponding Lagrangian function is given by
\begin{equation} \label{eq:lagrange}   % for double column
\begin{split}
  \mathcal{L} &= \sum_{\tau=1}^T \hspace{-1mm} \Big(w_1 D_{1\tau} + w_2 D_{2\tau}\\
   & +\lambda_{1\tau} ( -D_{1\tau}+(\bar{\eta}+\eta2^{-2r_{2\tau}})2^{-2r_{1\tau}}) \\
   &+\lambda_{2\tau} ( -D_{2\tau}+(\bar{\eta}+\eta2^{-2r_{1\tau}})2^{-2r_{2\tau}}) \\
   & + \lambda_{3\tau} \big(-\log_2 D_{1\tau} -\log_2  D_{2\tau} \\
  &+ \log_2(\bar{\eta}+\eta 2^{-2(r_{1\tau}+r_{2\tau})}) - 2(r_{1\tau}+r_{2\tau})\big )\\
  &  +  \lambda_{4\tau} ((D_{1\tau}-\bar{\eta})2^{2r_{2\tau}} -\eta) \\
  & +  \lambda_{5\tau} ((D_{2\tau}-\bar{\eta})2^{2r_{1\tau}} -\eta)  \\
  & + \mu_{1\tau} \sum_{i=1}^\tau \frac{1}{h_1}(2^{2r_{1i}}-1)
            +\mu_{2\tau} \sum_{i=1}^\tau \frac{1}{h_2}(2^{2r_{2i}}-1)   \\
  & -\theta_{1\tau}r_{1\tau}-\theta_{2\tau}r_{2\tau} \Big),
\end{split}
\end{equation}
%\begin{equation} \label{eq:lagrange}
%\begin{split}
%  \mathcal{L} &= \sum_{\tau=1}^T \hspace{-1mm} \Big(w_1 D_{1\tau} + w_2 D_{2\tau}
%    +\lambda_{1\tau} ( -D_{1\tau}+(\bar{\eta}+\eta2^{-2r_{2\tau}})2^{-2r_{1\tau}})
%   +\lambda_{2\tau} ( -D_{2\tau}+(\bar{\eta}+\eta2^{-2r_{1\tau}})2^{-2r_{2\tau}}) \\
%   & + \lambda_{3\tau} \big(-\log_2 D_{1\tau} -\log_2  D_{2\tau}
%  + \log_2(\bar{\eta}+\eta 2^{-2(r_{1\tau}+r_{2\tau})}) - 2(r_{1\tau}+r_{2\tau})\big )
%    +  \lambda_{4\tau} ((D_{1\tau}-\bar{\eta})2^{2r_{2\tau}} -\eta) \\
%  & +  \lambda_{5\tau} ((D_{2\tau}-\bar{\eta})2^{2r_{1\tau}} -\eta)
%   + \mu_{1\tau} \sum_{i=1}^\tau \frac{1}{h_1}(2^{2r_{1i}}-1)
%            +\mu_{2\tau} \sum_{i=1}^\tau \frac{1}{h_2}(2^{2r_{2i}}-1)
%   -\theta_{1\tau}r_{1\tau}-\theta_{2\tau}r_{2\tau} \Big),
%\end{split}
%\end{equation}
where $\lambda_{i\tau}$, $\mu_{k\tau}$, and $\theta_{k\tau}~ (i = 1, 2, \cdots, 5, k=1,2)$ are non-negative Lagrange multipliers associated with constraints (\ref{cst:lambda1})--(\ref{cst:theta}), respectively.

Next, we investigate the property of the multipliers to simplify the optimization problem.

\begin{proposition} \label{prop:lambda0}
    $\lambda_{1\tau}=0$, $\lambda_{4\tau}=0$, $\lambda_{5\tau}=0$, and $\lambda_{3\tau}>0$ for $\tau=1,\cdots,T$.
\end{proposition}

\begin{proof}
See Appendix \ref{prf:prop1}.
\end{proof}

The following proposition specifies $\lambda_{2\tau}$.

\begin{proposition} \label{prop:lambda2}
     $\lambda_{2\tau}=0$ if $r_{2\tau}\geq g(r_{1\tau})$, and $\lambda_{2\tau}>0$ if $r_{2\tau}< g(r_{1\tau})$, where $g(r)=-\frac12 \log_2 \frac{w_1\bar{\eta}2^{-2r}} {w_2(\bar{\eta}+\eta2^{-2r})^2-w_1\eta 2^{-4r}}$.
\end{proposition}

\begin{proof}
See Appendix \ref{prf:prop2}.
\end{proof}

In the sequel, we denote $x=2^{-2r_{1\tau}}$ and $y=2^{-2r_{2\tau}}$ for  simplicity.
    We also define
\begin{equation}\label{df:wl}
  \nu_{k\tau}=\frac{1}{\sum_{i=\tau}^T \mu_{ki}}, \qquad k=1,2, \tau=1,\cdots, T,
\end{equation}
 as the \textit{water level} associated with transmit power of node $k$.

\subsection{Structure of Optimal Policy}
From our discussion so far, we have the following observations on the structure of the optimal power allocation.

\begin{theorem} \label{th:p_structure}
    For the optimal offline power allocation, following conditions should be satisfied:
    \begin{enumerate}
      \item the energy buffer of node $k$ should be depleted if node $k$ harvests more energy on average in future slots;
      \item transmit power $p_{k\tau}$ should be increased after the slots in which the energy buffer of node $k$ is depleted;
      \item transmit power $p_{k\tau}$ should be decreased after the slots in which the energy buffer of the other node $\tilde{k}$  is depleted.
    \end{enumerate}
\end{theorem}

\begin{proof}
    See Appendix \ref{prf:thm1}.
\end{proof}

Theorem \ref{th:p_structure} presents the rule of the optimal power control in the offline case, which is obtained from the causality constraint of energy arrivals and the interplay between the two nodes.
    In the traditional power allocations under average power constraint, water level is the inverse of the first order derivative of the objective function and the system performance would be optimal if the water level is even throughout the transmission.
To be specific, under the optimal power allocation, allocating some small amount of additional power in whatever way leads to the same performance improvement, i.e., the marginal performance gain is even \cite{Ulukus-2012-packet_mac, GWF-2006}.
    Whereas, in the energy harvesting scenario, the harvested energy can never flow from future to the past so that the water level cannot be made even throughout the transmission.
Therefore, when the energy harvested by a node is larger than the previous period, we can only increase the transmit power of following period, not spare some energy for the previous period.
    Note that this also increases the water level of the following period.
%On the contrary, in a period in which the node harvests more energy in the earlier part than in the later part, we can maintain a constant transmit power and a constant water level throughout the period without violating the energy causality constraint.
    % In order to minimize the distortion, the node should use up all available energy (depleting the energy buffer) by the end of this period.
Therefore, the water level of each node is piecewise constant and monotonically increasing.

    Moreover, during a period when the water level of node $k$ is constant and the transmit power of the other node $\tilde{k}$ is increased, the weighted-sum distortion tends to be smaller and the water level of  both nodes of the following slot becomes larger.
It should be noted, however, that the best performance is achieved when the water-level is even throughout this period.
    Thus, we should use smaller transmit power in the following slots and increase the transmit power of previous slots, so that the water level could be even during this period.

\begin{example}
    Consider the special case of $\eta=1$ and the two sources are perfectly correlated, namely, transmitting the same source using two energy harvesting nodes.
        In this case, the  rate-distortion region \eqref{eq:cstr6}--\eqref{eq:cstr8} is degraded to
        \begin{eqnarray*}
                D_{1\tau} \hspace{-2.75mm} &\geq& \hspace{-2.75mm} 2^{-2(r_{1\tau}+r_{2\tau})}, \\
                D_{2\tau} \hspace{-2.75mm} &\geq& \hspace{-2.75mm} 2^{-2(r_{1\tau}+r_{2\tau})}.
        \end{eqnarray*}
        It is clear that the minimum weighted-sum distortion is obtained when $D_1=D_2=2^{-2(r_1+r_2)}$.
            Therefore, the optimal power control can be obtained by solving the following problem:
            \begin{eqnarray*}
                    \nonumber
                    \hspace{-6mm}  \mathop{\mathrm{minimize~}}\limits_{ r_{k\tau}} ~~~\hspace{-3mm}& &  \hspace{-3mm}\sum_{\tau=1}^T 2^{-2(r_{1\tau}+r_{2\tau})} \\
                    \hspace{-6mm} (\mathrm{P}'_3)~~\mathrm{subject~to} \quad \hspace{-3mm}&&\hspace{-3mm}  \sum_{i=1}^\tau \frac{1}{h_k}(2^{2r_{ki}}-1) - \sum_{i=1}^\tau e_{ki} \leq 0;\\
                     \quad \hspace{-3mm}&&\hspace{-3mm} -r_{k\tau} \leq 0,      \forall~k=1,2,~ \tau=1,\cdots,T.
                \end{eqnarray*}

            By defining corresponding Lagrangian and set the derivatives with respect to $r_{1\tau}$ and $r_{2\tau}$ to zero, we have
             \begin{eqnarray*}
                p_{1\tau} \hspace{-2.75mm} &=& \hspace{-2.75mm} \frac{1}{h_1} \left[  \Big( \frac{h_1^2\nu^2_{1\tau}}{h_2\nu_{2\tau}}\Big)^{\frac13}-1 \right]^+, \\
                p_{2\tau} \hspace{-2.75mm} &=& \hspace{-2.75mm} \frac{1}{h_2} \left[  \Big( \frac{h_2^2\nu^2_{2\tau}}{h_1\nu_{1\tau}}\Big)^{\frac13}-1 \right]^+,
         \end{eqnarray*}
         where $\nu_{k\tau}=\frac{1}{\sum_{i=\tau}^T \mu_{k\tau}}$ is the water level (see \eqref{df:wl}), $[x]^+=x$  if $x>0$ and $[x]^+=0$ otherwise.

         It is clear that $p_{1\tau}$ is increasing with $\nu_{1\tau}$ and decreasing with $\nu_{2\tau}$.
            Since $\nu_{1\tau}$ is increasing with time and changes only when the energy buffer of  node 1 is depleted, we know $p_{1\tau}$ should be increased when its own energy buffer is depleted and should be decreased when the energy buffer of node 2 is depleted, as shown in Fig. \ref{fig:eta_sigma}.
         Similar conclusion can be drawn for the transmit power of node 2, which validates our result in Theorem \ref{th:p_structure}.
            Moreover, this example implies that the correlation between  two sources does not change the structure of the optimal power control, even in the extreme cases such as $\eta=1$ and $\eta=0$.
\end{example}

{
\centering
\begin{figure}[!t]
\includegraphics[width=3.8in]{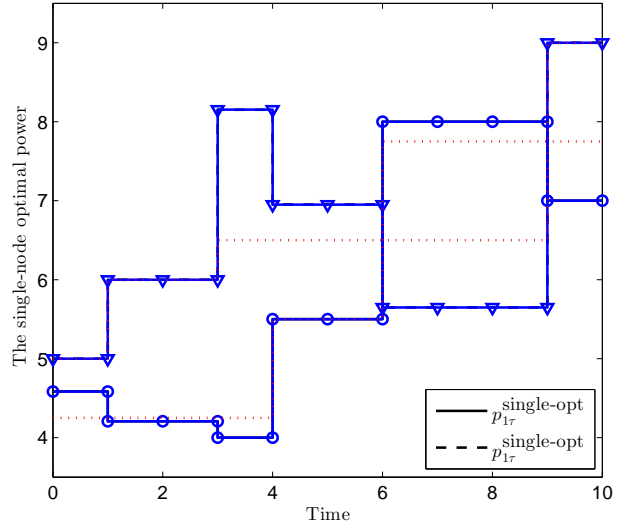}
\caption{The optimal power control of the two nodes, where $\eta=1$. } \label{fig:eta_sigma}
\end{figure}
}

\subsection{Iterative Solution}
In this subsection, we present a power allocation algorithm to find the optimal policy efficiently, as shown in Table \ref{Alg:iter_wt_fil}.

First, we set the transmit power of node 2 to zero, i.e., $p_{2\tau}^0=0$ for $1\leq\tau\leq T$ and consider a single node  distortion minimization for node 1.
    In this case, we have $D_1=2^{-2r_{1\tau}}$ and $D_2=\bar{\eta}+\eta2^{-2r_{1\tau}}$.
Recalling that $r_{1\tau}=\frac12\log_2(1+h_1p_{1\tau})$, $(\mathrm{P}_3)$ can be convert to:
\begin{eqnarray*}
 \mathop{\mathrm{minimize~}}\limits_{p_{1\tau}} ~~~\hspace{-3mm}& &  \hspace{-2mm}\sum_{\tau=1}^T \frac{w_1+w_2\eta}{1+h_1p_{1\tau}} +w_2\bar{\eta}  \\
(\mathrm{P}_4)~~~\mathrm{subject~to} \quad \hspace{-3mm}&&\hspace{-2mm} \sum_{i=1}^\tau p_{1i} \leq \sum_{i=1}^\tau e_{1i}, \\
                        \hspace{-2mm}&&\hspace{-2mm}  \quad~  p_{k\tau} \geq 0,\quad   k=1,2,~\tau=1,\cdots,T.
\end{eqnarray*}

Since this is a  convex problem with causal energy constraint, its solution can be obtained using  directional water-filling \cite{Ulukus-2011-policy}.
    In essence, this strategy tries to allocate energy as even as possible throughout the transmission.
To be specific, the strategy divides the period of transmission into $K$  bands, where the $j$-th band starts from slot  $L_{j}+1$ and ends with slot $L_{j+1}$, i.e., $L_{j}+1\leq \tau \leq L_{j+1}, j=1,\cdots, K$.
    We denote $e(0)=0$, $e_1(\tau)=\sum_{i=1}^\tau e_{1i}$,  $L_0=0$, and $L_{K}=T$. Then the remaining $L_j$ is determined by
 \begin{equation} \label{eq:band_dtm}
   L_j=\arg \min\limits_{{L_{j-1}+1\leq \tau \leq T}} \frac{e(\tau)-e(L_{j-1})}{\tau-L_{j-1}}, ~\textrm{for}~ 1\leq j \leq K-1.
 \end{equation}

 In each band, the transmit power is the same for each slot,
 \begin{equation} \label{eq:pkt_dtm}
   p_{1\tau}^0 = \frac{e(L_j)-e(L_{j-1})}{L_j-L_{j-1}}, ~ L_{j-1}+1\leq \tau \leq L_j.
 \end{equation}
Using $\{p_{1\tau}^0\}$ as the initial power allocation, $(\mathrm{P}_3)$ can be solved %by an iterative backward water-filling
iteratively (see Table \ref{Alg:iter_wt_fil}).
    To be specific, in the  $l$-th  iteration, we  solve the optimal $\{p_{2\tau}^l\}$  based on previous output $\{p_{1\tau}^{l-1}\}$,  and then solve $\{p_{1\tau}^l\}$  based on $\{p_{2\tau}^l\}$.

\begin{algorithm}[!t]
\algsetup{linenosize=\small}
\scriptsize
\caption{The iterative generalized backward-water-filling}
\begin{algorithmic}[1]\label{Alg:iter_wt_fil}
\REQUIRE~~\\%Initialization
\STATE Set $l=1$, $\varepsilon = 10^{-6}$;
\STATE Set $p_{2\tau}^0=0$ for $\tau=1,\cdots,T$;
\STATE $\verb"/"*$\textit{Determine the initial transmit power}$*\verb"/"$
\STATE Set~~~ $L_0=0$, $L_{K}=T$.
\STATE  Solve $L_j$ for $j=1, \cdots,K-1$ using (\ref{eq:band_dtm});
\STATE  Solve $p_{1\tau}^0$ for $\tau=1,\cdots,T$ using (\ref{eq:pkt_dtm});

\ENSURE~~\\%Iteration
%\STATE $\verb"/"*$ \textit{Iterative solution using generalized backward- water-filling} $*\verb"/"$
\WHILE {$\Delta P>\varepsilon$}
\STATE $\verb"/"*$\textit{Solve $\{p_{2\tau}^l\}$ using $\{p_{1\tau}^{l-1}\}$, then $\{p_{1\tau}^l\}$ using $\{p_{2\tau}^l\}$}$*\verb"/"$
    \FOR {$k=2:-1:1$}
    \STATE Set $p_{kT}^l=e_{kT}$ and calculate water level $\zeta_{kT}(e_{kT})$;
        \FOR{$\tau =T-1 :-1: 1$}
            \STATE Set $p_{k\tau}^l=e_{k\tau}$ and calculate water level $\zeta_{k\tau}(e_{k\tau})$;
            \IF {$\zeta_{k\tau}(e_{k\tau})>\zeta_{k(\tau+1)}(p_{k(\tau+1)}^l)$}
                \STATE Find the number $m$ of slots where $\zeta_{k\tau}(e_{k\tau}) > \zeta_{ks}(p_{ks})$ for $s=\tau+1,\cdots,T$;
                \STATE Find the energy $\hat{e}_{k\tau}$ such that $\zeta_{k\tau}(\hat{e}_{k\tau}) = \zeta_{ks}(p_{k(\tau+1)}^l)$;
                \STATE Pour remaining energy $\check{e}_{k\tau}=e_{k\tau}-\hat{e}_{k\tau}$ over slots $[\tau,...,\tau+m]$;
                \STATE Update $p_{ks}^l$ and $\zeta_{ks}(p_{ks}^l)$ for $s=\tau,\cdots,T$;
              \ENDIF
        \ENDFOR
   \ENDFOR
\STATE  Calculate error $\Delta P=\sum_{k=1}^2\sum_{\tau=1}^T (|p_{k\tau}^l-p_{k\tau}^{l-1}|)$;
\STATE  Update iteration index $l=l+1$;
\ENDWHILE
\STATE \textbf{Output:} $\{p_{1\tau}^l, p_{2\tau}^l\}_{(\tau=1,\cdots,T)}$.
\end{algorithmic}
\end{algorithm}

If we differentiate Lagrangian \eqref{eq:lagrange}  with respect to  $r_{k\tau} $ and set it to zero, we have
\begin{eqnarray}
\label{eq:water_1}
        h_1 x\left (\lambda_{2\tau}\eta xy + \frac{\lambda_{3\tau}\eta xy} {\bar{\eta}+\eta xy} + \frac{\lambda_{3\tau}}{\ln 2} \right)
   \hspace{-2.75mm} &\leq & \hspace{-2.75mm} \frac{1}{\nu_{1\tau}},\\
\label{eq:water_2}
        h_2 y\left (\lambda_{2\tau}\bar{\eta}y+\lambda_{2\tau}\eta xy + \frac{\lambda_{3\tau}\eta xy} {\bar{\eta}+\eta xy} + \frac{\lambda_{3\tau}}{\ln 2} \right)
   \hspace{-2.75mm} &\leq & \hspace{-2.75mm} \frac{1}{\nu_{2\tau}}.
\end{eqnarray}

Define two generalized water levels \cite{Ulukus-2012-packet_mac, GWF-2006} $ \zeta_{1\tau}(p_{1\tau})$ and $\zeta_{2\tau}(p_{2\tau})$ as the inverse of the left-hand side of (\ref{eq:water_1}) and (\ref{eq:water_2}), respectively, and thus the two KKT conditions can be rewritten as
    \begin{eqnarray*}
    \zeta_{1\tau}(p_{1\tau})\hspace{-2.75mm} &\geq& \hspace{-2.75mm} {\nu}_{1\tau}, \\
     \zeta_{2\tau}(p_{2\tau})\hspace{-2.75mm} &\geq& \hspace{-2.75mm} {\nu}_{2\tau},
    \end{eqnarray*}
 for $\tau=1,\cdots,T$.
% Note that $\tilde{\nu}_{k\tau}$ is in general different from $\nu_{k\tau}$.
%    Nevertheless, $\tilde{\nu}_{k\tau}$ will eventually converge to  $\nu_{k\tau}$

 Based on $\{p_{1\tau}^{l-1}\}$, the optimal power allocation for node 2 can be solved using the generalized backward water-filling process.
    We first pour the harvested energy $e_{2T}$ into the $T$-th slot.
 Hence, the transmit power would be $p_{2T}^l=e_{2T}$ and the water level $\zeta_{2T}^l(e_{2T})$ can be calculated based on $p_{2T}^l$ and $p_{1T}^{l-1}$.
    Next, we fill $e_{2(T-1)}$ over the $(T-1)$-th slot until the harvested energy $e_{T-1}$ is depleted or until the water level reaches $\zeta_{2T}^l(e_{2T})$.
 When the former case happens, $p_{2T}^l$ remains unchanged and $p_{2(T-1)}^l$ can be calculated directly.
    If the latter case happens, the remaining energy will be evenly filled over slots $[T-1, T]$.
        Afterwards,  the transmit power $p_{2(T-1)}^l$ can be calculated and transmit power $p_{2T}^l$ would be updated.
In addition, the water levels of the two slots can  also be updated accordingly.
    By repeating this process until the energy harvested in the first slot is filled,  the optimal $\{p_{2\tau}^l\}$ can be obtained.

 Likewise, the optimal $\{p_{1\tau}^l\}$  can be obtained based on $\{p_{2\tau}^l\}$.
    By repeating this process until the difference between the outputs of two adjacent iterations is negligible, i.e., the predefined threshold is reached, we will finally obtain the optimal power allocation for both nodes.

\section{Online Power Allocation} \label{sec:4online}
 {For the online case, only the distribution of the energy arrivals is known a priori, and thus the nodes cannot optimize the whole transmission process in advance.
    In order to minimize the weighted-sum distortion, each node needs to adjust its transmit power based on the energy status of the system in real-time.}
    Due to the stochastic nature of the energy harvesting process, the transmit power and the remaining energy of each node will also be random.
In this section, we investigate this causal case and analyze the  expected weighted-sum distortion.

\subsection{Problem Formulation}
    In this section, we normalize the harvested energy $e_{k\tau}$ and the transmit power $p_k$ using a constant $\delta$ and consider a set of discrete $e_{k\tau}$ and $p_k$, i.e., $e_{k\tau}\in \mathbb{Z}_{++}$ and $p_k\in \mathbb{Z}_{++}$.
It is clear that the quantized energy and power approach their original values when $\delta$ goes to zero.
    The normalized capacities of energy buffers are denoted as $L_1$ and $L_2$, respectively.
We say the system is in\textit{ energy state} $s$ if the remaining energy in the two buffers is $i=\lceil\frac{s}{L_2}\rceil$ and $j=s\,(\text{mod}\,L_2)+L_2 I_{s\, (\text{mod}\,L_2)=0}$, respectively, where $I_{A}$ is the indicator function (1 if $A$ is true and 0 otherwise).
    By denoting $L=L_1\times L_2$,  the energy state space would be $\Omega=[1,\cdots,L]$.

\begin{definition}
    An online \textit{power control function} $\rho$ is a mapping from the energy state space $\Omega$ to $\mathbb{Z}^2_{++}$.
        Give an energy state $s$,  $[\rho(s)]_k$ can be interpreted as the corresponding transmit power of node $k$, i.e.,  $p_1=[\rho(s)]_1$,  $p_2= [\rho(s)]_2$.
\end{definition}

Note that the a power control function $\rho$ is feasible only if the resulting transmit powers are positive integers and satisfy the energy constraint specified by energy state $s$.
    We denote the set of all feasible power control functions as
\begin{eqnarray*}\label{df:fsb_rho}
  \mathcal{F}_\rho \hspace{-3mm} &=& \hspace{-3mm}  \{ \rho | 1\leq \rho_1(s) \leq i, 1\leq \rho_2(s) \leq j,  \\
                                  \hspace{-3mm} && \hspace{-3mm} ~~~~~~~~~
                                  \rho_1(s) \in \mathbb{Z}_{++}, ~\rho_2(s) \in \mathbb{Z}_{++},  \forall s\in \Omega \}.
\end{eqnarray*}

Let $\rho_\tau$ be the control function for the $\tau$-th slot.
    The sequence $\varrho=\{\rho_1,\rho_2,\cdots\}$ of control functions is referred to as a \textit{power control policy}.
If the control function is the same for all slots, we call the policy a \textit{stationary power control policy}.
    In addition, we denote $\varrho_\tau=\{\rho_\tau,\rho_{\tau+1},\cdots\}$ for $\tau\geq2$ as a power control policy starting from slot $\tau$.

Let $S_\tau\in \Omega$ be the random energy state in slot $\tau$.
    Given $S_\tau=s$, we denote the corresponding minimum weighted-sum distortion under control function $\rho_\tau$ as $d_{\rho_\tau}(s)$.
That is,
\begin{equation} \label{df:d_rou_s}
    d_{\rho_\tau}(s)=D(\boldsymbol{p}),
\end{equation}
where $D(\boldsymbol{p})$ is defined in \eqref{df:w_dis} and  $\boldsymbol{p}=[p_1,p_2]=[[\rho(s)]_1,[\rho(s)]_2]$ is the transmit power of the two nodes under control function $\rho$.
    In addition, we denote $\boldsymbol{d}_{\rho_\tau}=[d_{\rho_\tau}(1),\cdots,d_{\rho_\tau}(L)]^{\textrm{T}}$ as the distortion vector under $\rho_\tau$.

Given the distribution of $e_{1\tau}$ and $e_{2\tau}$ and the power control function $\rho_\tau$, we denote the transfer probability from state $s$ to state $t$ as $q_{st}=\Pr\{i\rightarrow i'\}\Pr\{j\rightarrow j'\}$, where $i'=\lceil\frac{t}{L_2}\rceil$ and $j'=t\,(\text{mod}\, L_2)+L_2 I_{t\,(\text{mod}\, L_2)=0}$.
    We  denote the corresponding probability transfer matrix as  $\textbf{P}_{\rho_\tau}$.

In this section, we investigate the online power control policy that minimizes the expected weighted-sum distortion in  $(\mathrm{P}_2)$.
    Since the objective function of  $(\mathrm{P}_2)$ is not tractable due to the complexity of the distortion region, we shall propose a cost function $v_\varrho(s)$ to characterize the distortion in the next subsection.
 {As will be shown in Theorem} \ref{th:3cost_equal},  {the expectation of $v_\varrho(s)$ equals the expected weighted-sum distortion.
    Therefore, the cost function $v_\varrho(s)$ is a reasonable metric for the system distortion for the online case.
Furthermore, since the cost function is defined as the weighted sum of current cost and expected future cost, one can expect that $(\mathrm{P}_2)$ may be solved by some stationary and convergent iterative process.
    }

\subsection{Cost Function}
Given the initial energy state $S_0=s$, we define $v_{\varrho}(s)$ as the cost associated with energy state $s$ and policy $\varrho$.
\begin{definition}
    The cost $v_{\varrho}(s)$ is a mapping from  energy state space $\Omega$ to $\mathbb{Z}_{++}$.
        To be specific, $v_{\varrho}(s)$ is the weighted sum of current distortion and the expectation of future distortion,
    \begin{equation} \label{df:cost}
        v_{\varrho}(s)=\bar\alpha d_{\rho_1}(s) + \alpha\mathbb{E}(v_{\varrho_2}(t)),
    \end{equation}
    where $0<\alpha<1$ is a weighting coefficient, $\bar\alpha=1-\alpha$, and
    \begin{equation} \label{eq:expt_cost}
        \mathbb{E}(v_{\varrho_2}(t))=\sum_{t=1}^L q_{st} v_{\varrho_2}(t)
    \end{equation}
         is the expectation of future cost. In addition, $\mathbb{E}(\cdot)$ is the expectation operator with respect to the randomness of energy harvesting process.
\end{definition}

For a feasible power control policy $\varrho$, since the resulting transmit power $\boldsymbol{p_{\tau}}$ is positive, the corresponding distortion $D(\boldsymbol{p_{\tau}})$ would be finite.
    Thus, the cost $v_{\varrho_\tau}(S_\tau) $ is also finite for each slot.
 {Since the distribution of energy harvesting process is known to each node, the transfer probability $q_{st}$ from an energy state $s$ to another energy state $t$ can be readily calculated.
    However, the cost function $v_{\varrho}(s)$ given by} \eqref{df:cost}  {still cannot be calculated directly since the average future cost $\mathbb{E}(v_{\varrho_2}(t))$ is unknown.
Nevertheless, we will establish a tractable analytic framework based on cost function and develop an iterative algorithm} (see Table \ref{alg:online})  {using current cost only.
    To be specific, by using current cost as an estimation of expected future cost, the optimal power control function for current cost can be determined} (see \eqref{eq:T_explains}).
 {Using this power control function, current cost will be updated according to} \eqref{eq:T_explains}.
     {By repeating the process of solving for the power control function and then applying the power control function iteratively,  the current cost will eventually converge to the minimum system cost and the corresponding power control function would minimize the expected weighted-sum distortion of the system.}

We denote $\boldsymbol{v}_\varrho=[v_{\varrho}(1),\cdots,v_{\varrho}(L)]^{\text{T}}$ as the cost vector.
    It can be seen that $\boldsymbol{v}_\varrho$ can be expressed in the following matrix form:
%\begin{eqnarray}\label{df:v_pi_0}
%    \boldsymbol{v}_\varrho \hspace{-3mm} &=& \hspace{-3mm} \bar\alpha \boldsymbol{d}_{\rho_1}
%                                                                            +\alpha \mathbb{E}(\boldsymbol{v}_{\varrho_2}) \\
%            \nonumber
%            \hspace{-3mm} &=& \hspace{-3mm}\bar\alpha \boldsymbol{d}_{\rho_1}
%                                                                            +\alpha \textbf{P}_{\rho_1}\boldsymbol{v}_{\varrho_2} \\
%            \nonumber
%            \hspace{-3mm} &=& \hspace{-3mm}\lim_{T\rightarrow\infty} \sum_{\tau=1}^T \bar\alpha \alpha^{\tau-1}
%                                                                          \textbf{P}_{\rho_0} \cdots\textbf{P}_{\rho_{\tau-1}}\boldsymbol{d}_{\rho_{\tau}}
%                                                                            +\alpha^T\textbf{P}_{\rho_{1}} \cdots\textbf{P}_{\rho_{T}}\boldsymbol{v}_{\varrho_{T+1}} ,
%\end{eqnarray}
\begin{eqnarray}\label{df:v_pi_0}
    \boldsymbol{v}_\varrho \hspace{-3mm} &=& \hspace{-3mm} \bar\alpha \boldsymbol{d}_{\rho_1}
                                                                            +\alpha \mathbb{E}(\boldsymbol{v}_{\varrho_2})
            =\bar\alpha \boldsymbol{d}_{\rho_1}
                                                                            +\alpha \textbf{P}_{\rho_1}\boldsymbol{v}_{\varrho_2} \\ \nonumber
            \hspace{-3mm} &=& \hspace{-3mm} \lim_{T\rightarrow\infty} \sum_{\tau=1}^T \bar\alpha \alpha^{\tau-1}
                                                                          \textbf{P}_{\rho_0} \cdots\textbf{P}_{\rho_{\tau-1}}\boldsymbol{d}_{\rho_{\tau}}
                                                                            +\alpha^T\textbf{P}_{\rho_{1}} \cdots\textbf{P}_{\rho_{T}}\boldsymbol{v}_{\varrho_{T+1}} ,
\end{eqnarray}
where $\textbf{P}_{\rho_{0}}$ is a unit matrix.

On one hand, if $\alpha$ is very small, we have $\boldsymbol{v}_\varrho \thickapprox \boldsymbol{d}_{\rho_1}$. In this case, the cost function focuses on current distortion and hence is minimized by the greedy power allocation policy.
    On the other hand, if $\alpha$ approaches unity, the cost function reduces to the expectation of future distortion, which is equal to the expected weighted-sum distortion associated with initial state vector $\boldsymbol{s}_0=[1,2,\cdots, L]$.

Also, note that the proposed cost function is different from the discounted cost model (without the item weighted by $\bar\alpha$) in the MDP theory in that the expected cost is equal to the expected distortion.
        The proposed cost function is also different  from the average cost without discounting of the MDP theory, which is much more difficult to deal with \cite{Pturmsn-2014-MDP}.

\begin{theorem} \label{th:3cost_equal}
    In a period of $T$ slots, for any power control policy $\varrho=\{\rho_1,\rho_2,\cdots, \rho_T\}$, we have
    \begin{equation*}\label{rt:prop_av_eqal}
               \mathbb{E}\left( \frac1T \sum_{\tau=1}^T v_{\varrho_\tau} (S_\tau) \right)
            = \mathbb{E}\left( \frac1T \sum_{\tau=1}^T  d_{\rho_\tau}(S_\tau) \right).
    \end{equation*}
\end{theorem}

\begin{proof}
    See Appendix \ref{prf:3cost_equal}
\end{proof}

Note that $d_{\rho_\tau}(S_\tau)$ equals the weighted-sum distortion in the $\tau$-th slot (see \eqref{df:d_rou_s}).
    Thus, we can solve $(\mathrm{P}_2)$ by dealing with the expected cost instead.

\subsection{Minimum Expected Cost} \label{subsec:4c}
In this subsection, we solve the expected cost minimization problem and show that the optimal power control can be obtained by iteratively applying some simple function to an arbitrary non-zero initial cost vector.

We denote the minimum cost vector as
\begin{equation}\label{df:v_opt}
  \boldsymbol{v}^* = \inf_{\varrho} \boldsymbol{v}_\varrho.
\end{equation}
That is to say, starting from an initial energy state $s$, $v^*(s)$ is the smallest cost among all achievable costs.

A policy $\varrho^*$  is said to be $\alpha$-optimal if
\begin{equation*}\label{df:pi_opt}
  \boldsymbol{v}_{\varrho^*} = \boldsymbol{v}^*.
\end{equation*}
That is, the cost under policy $\varrho^*$ is the minimum cost $\boldsymbol{v}^*$.

For an $L$ dimensional vector $\boldsymbol{v}$, we define a mapping from $\mathbb{R}_{++}^L$ to $\mathbb{R}_{++}^L$:
\begin{eqnarray} \label{df:T}
   \mathbb{T} (\boldsymbol{v}) =\min_{\rho}\left\{\bar\alpha \boldsymbol{d}_\rho + \alpha \textbf{P}_\rho \boldsymbol{v}\right\}
\end{eqnarray}
where the minimization is performed for each element of $\boldsymbol{v} $.
    That is, $\mathbb{T}(\boldsymbol{v})(s)$ maps the $s$-th element $\boldsymbol{v}(s)$ to
    \begin{equation}\label{eq:T_explains}
        \mathbb{T}(\boldsymbol{v})(s) = \min_{\rho}\left\{\bar\alpha d_\rho(s) + \alpha \sum_{t=1}^{L} q_{st} v(t)\right\}.
    \end{equation}

Particularly, the mapping $\mathbb{T}(\boldsymbol{v})$ has the following property.
\begin{theorem} \label{prop:contract4}
    $\mathbb{T}(\boldsymbol{v})$ is a contraction mapping under the maximum norm $\Vert \cdot \Vert_\infty$.
        That is, for any $\boldsymbol{u}\in \mathbb{R}_{++}^L$ and $\boldsymbol{v}\in \mathbb{R}_{++}^L$, we have
        \begin{equation*}\label{rt:prop2}
            \Vert  \mathbb{T}(\boldsymbol{u})-  \mathbb{T}(\boldsymbol{v})\Vert_\infty
            \leq \beta \Vert \boldsymbol{u}-\boldsymbol{v}\Vert_\infty
        \end{equation*}
        for some constant $0<\beta<1$.
\end{theorem}

\begin{proof}
    See Appendix \ref{prf:contract4}.
\end{proof}

Moreover, the convergence of applying a contraction mapping iteratively is guaranteed by the following theorem \cite{Banach-1922}.
\begin{theorem} \label{th:fixpoint}
    (\textit{Banach's Fixed Point Theorem}) Let $(X, d)$ be a non-empty complete metric space with a contraction mapping $\mathbb{T}:X\rightarrow X$.
        Then $\mathbb{T}$ admits a unique fixed-point $x^*$ in $X$  (i.e. $\mathbb{T}(x^*)=x^*$).
    Furthermore,  $x^*$ can be found as follows: start with an arbitrary element  $x_0\in X$ and define a sequence ${x_n}$ by $x_n=\mathbb{T}(x_{n-1})$, then $x_n\rightarrow x^*$.
\end{theorem}

Based on our previous analysis and Banach's fixed point theorem, we have the following theorem on the minimum cost vector $\boldsymbol{v}^*$.
\begin{theorem} \label{th:main4}
    For the minimum cost vector $\boldsymbol{v}^*$, following properties hold ture:
    \begin{enumerate}
     \item $\boldsymbol{v}^*$ is the fixed point of mapping $\mathbb{T}(\boldsymbol{v})$, i.e., $\mathbb{T}(\boldsymbol{v}^*)= \boldsymbol{v}^*$;
     \item for any positive $\boldsymbol{v}_0\in \mathbb{R}_{++}^{L}$, $\lim_{T\rightarrow\infty}\mathbb{T}^T(\boldsymbol{v}_0) =\boldsymbol{v}^*$.
     \end{enumerate}
\end{theorem}

\begin{proof}
    See Appendix \ref{prf:main4}.
\end{proof}

Therefore, the minimum cost vector $\boldsymbol{v}^*$ can be obtained by simply applying $\mathbb{T}(\boldsymbol{v})$ to an arbitrary positive initial vector $\boldsymbol{v}_0$ iteratively.
    This also means that the $\alpha$-optimal power control policy can be chosen as a \textit{stationary policy} $\varrho^*=\{\rho^*,\rho^*,\cdots\}$, where $\rho$ is solved from $\mathbb{T}_{\rho^*}(\boldsymbol{v}^*)= \boldsymbol{v}^*$ (see \eqref{df:T}).
The algorithm is summarized in  Table \ref{alg:online}.

By using the algorithm in Table 2, we can find the optimal online power control without directly calculating either the information theoretic distortion or the cost function (see \eqref{df:cost}).
    Instead, we start from an arbitrary non-negative initial cost vector and solve the optimal power control function for the current cost, and then simply repeat this operation until the output cost of two adjacent iterations is negligible.
According to Theorem \ref{th:main4}, the obtained power control function will minimize the expected cost, which is equal to the expected weighted distortion.
    By measuring its own remaining energy and inquiring the remaining energy of the other node in each slot, each node can determine its transmit power in real-time.
Also, we note that the proposed cost function based approach can also be applied to the online scheduling for other networks.

\begin{algorithm}[!t]
\algsetup{linenosize=\small}
\scriptsize
\caption{The iterative solution for online power control}
\begin{algorithmic}[1]\label{alg:online}
\REQUIRE~~\\%Initialization
\STATE Set $l=1$, $\varepsilon = 10^{-3}$;
\STATE Set $\boldsymbol{v}^0=\textbf{0}$;

\ENSURE~~\\%Iteration
\WHILE {$\Delta \boldsymbol{v}>\varepsilon$}
\FOR {$s=1:1:L$}
        \STATE Search transmit power $p_1$, $p_2$ (i.e., optimal policy $\rho^l$ for state $s$) using \eqref{eq:T_explains};
    \ENDFOR
    \STATE Calculate distortion $\boldsymbol{d}_{\rho^l}$ using \eqref{rt:D_r} and probability transfer matrix $\textbf{P}_{\rho^l}$;
    \STATE Update $\boldsymbol{v}^l$ using \eqref{df:cost};
\STATE  Calculate error $\Delta \boldsymbol{v}=\sum_{s=1}^L |\boldsymbol{v}^l(s)-\boldsymbol{v}^{l-1}(s)|$;
\STATE  Update iteration index $l=l+1$;
\ENDWHILE
\STATE \textbf{Output:} $\boldsymbol{v}^*=\boldsymbol{v}^l, \rho^*=\rho^l$.
\end{algorithmic}
\end{algorithm}

After $\rho^*$ has been obtained, one can calculate the corresponding probability transfer matrix $\textbf{P}^*$ accordingly.
    Thus, the stationary distribution $\boldsymbol{\pi}$ of the energy states is \cite{Fan-SCT}
\begin{equation} \label{rt:pi_stationary}
    \boldsymbol{\pi}^*= \textbf{1}(\textbf{I}-\textbf{P}^*+\Theta)^{-1},
\end{equation}
where $\textbf{1}_{1\times L}$ is a vector of ones and $\Theta_{L\times L}$ is a matrix of ones.

%\begin{remark}
%    Theorem \ref{th:main4} shows that the solution can be obtained by simple iterations.
%        Moreover, the optimal power control policy is a stationary policy.
%\end{remark}

For the stationary power control policy $\varrho^*$, we have
    \begin{eqnarray}
    \label{dr:final_rt_1}
    \nonumber
           \hspace{-7mm} &&\mathbb{E}\left( \lim_{T\rightarrow\infty}\frac1T \sum_{\tau=1}^T v_{\varrho^*} (S_\tau) \right) = \mathbb{E}\left( v_{\varrho^*} (S_\tau) \right)
           = \sum_{s=1}^L \pi_{s}^* v_{\varrho^*}(s) .
    \end{eqnarray}

According to Theorem  \ref{th:3cost_equal}, we have $\mathbb{E}\left( \frac1T \sum_{\tau=1}^T v_{\varrho^*} (S_\tau) \right) =\mathbb{E}\left( \frac1T \sum_{\tau=1}^T  d_{\rho_\tau}(S_\tau) \right)$, and hence
       \begin{eqnarray}
    \mathbb{E}\left( \lim_{T\rightarrow\infty}\frac1T \sum_{\tau=1}^T  d_{\rho_\tau}(S_\tau) \right)=\sum_{s=1}^L \pi_{s}^* v_{\varrho^*}(s),
    \end{eqnarray}
   which would be the solution to the online power control problem $(\mathrm{P}_2)$ according to Theorem \ref{th:main4}.

\section{Numerical Results}\label{sec5:simulation}
In our simulations, we assume that the correlation coefficient between the samples of the two nodes is $\sqrt{\eta}=0.8367$ (i.e., $\eta=0.7$).
    The weighting coefficients for the distortions are $w_1=0.3$ and $w_2=0.7$.
The channel gain between the two nodes and the fusion center are $h_1=0.8$ and $h_2=0.5$, respectively.
    For simplicity, we set the slot length to $T_s=1$ s and the system bandwidth to $W=1$ Hz.

\subsection{Offline Power Control}
For the offline power control, we consider $T=10$ slots of transmissions.
    Assuming that the harvested energy at node 1 and node 2 are both integer random variables drawn uniformly from $\{1, 2, \cdots, 10\}$, we consider the following realization of energy arrivals in Joule (J):
    \begin{eqnarray*}
      \{e_{1\tau}\}\hspace{-2.75mm} &=& \hspace{-2.75mm} [5,     6,     2 ,    4,     9,     2,    10,     8,     6 ,    7], \\
      \{e_{2\tau}\}\hspace{-2.75mm} &=& \hspace{-2.75mm} [5,    10,     2,     9,    10,     9,     2,     4,     5,     9].
    \end{eqnarray*}

%%%%%%%%%%%%%%%%%%%%%%%%%%%%%%
\begin{figure}[!t]%[htp]   %cross column: add *

\hspace{-6 mm}
    \begin{tabular}{cc}
    \subfigure[Power allocations optimized for each single user]
    {
    \begin{minipage}[t]{0.5\textwidth}
    \centering
    {\includegraphics[width = 3.7in] {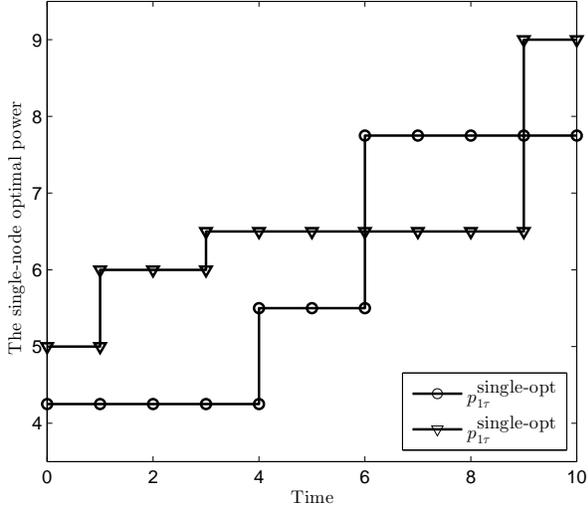} \label{fig:p120}}
    \end{minipage}
    }\\

    \hspace{0.01\textwidth}
    \subfigure[Power allocation by iterative generalized-backward water-filling]
    {
    \begin{minipage}[t]{0.5\textwidth}
    \centering
    {\includegraphics[width = 3.7in] {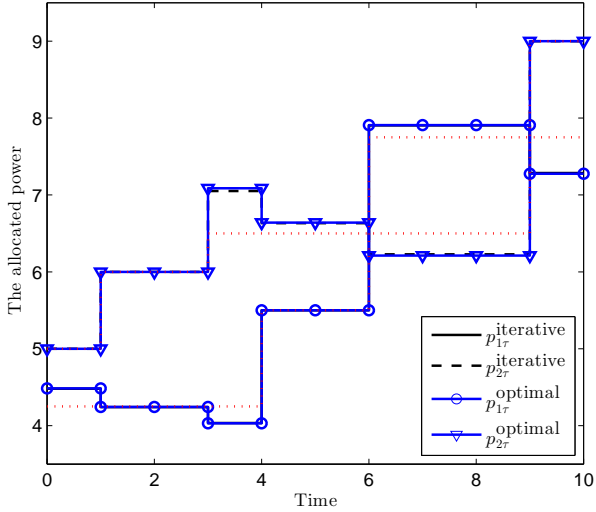} \label{fig:p12k}}
    \end{minipage}
    }
    \end{tabular}

\caption{The optimal offline power control.} \label{fig:power}
\end{figure}
%%%%%%%%%%%%%%%%%%%%%%%%%%%%%%

 The optimal power allocation $\{p_{kt}^{\textrm{single-opt}}\}$ for a single user, which are obtained by solving $(\mathrm{P}_3)$, are presented in Fig. \ref{fig:p120}.
    We observe that the transmit power of each node is constant within each band and it increases when the band is changing.
In particular, the energy buffer will be depleted in the last slot of each band, following by an increase in transmit power in the next band.
    It is worth noting that although more energy is harvested in the next band,  the newly harvested energy cannot help the transmission in previous bands due to the causality constraint.
%As is shown in Fig.~\ref{fig:p120}, the energy does not flow to the left to make the transmit power equal throughout the transmission.

The optimal offline power allocation $\{p_{kt}^{\textrm{iterative}}\}$  for node 1 (the solid curve) and node 2 (the dashed curve), which are obtained by the iterative generalized backward-water-filling algorithm (see Table \ref{Alg:iter_wt_fil}), are presented in Figure \ref{fig:p12k}.
    The power allocation $\{p_{kt}^{\textrm{single-opt}}\}$ is also plotted for reference (the dotted curves).
From Fig. \ref{fig:p12k},  we observe that during a period with constant $p_{kt}^{\textrm{single-opt}}$, $p_{kt}^{\textrm{iterative}}$ is decreasing and $p_{\tilde{k}t}^{\textrm{iterative}}$ is increasing.
    Intuitively, this result makes sense because in each slot where $p_{\tilde{k}t}^{\textrm{single-opt}}$ is increased,\footnote{This actually occurs when the energy buffer of node $\tilde{k}$ is emptied and $\nu_{\tilde{k}t}$ is increased.} $p_{kt}^{\textrm{iterative}}$ needs to be decreased so that water level $\zeta_{kt}$ can be constant in this period.
Note that when $p_{kt}^{\textrm{single-opt}}$ is constant, its influence to the other node $\tilde{k}$ does not change. Therefore, the power of node $\tilde{k}$ can be optimized as if it is in a single node system.
    These observations validates the results in Theorem \ref{th:p_structure}.
The optimal power allocations of the two nodes obtained by Matlab optimization solver are also shown by curves marked by circles and triangles, respectively.
    It can be seen that the results match well with the solution obtained by the iterative generalized backward-water-filling algorithm.  %Moreover, the proposed algorithm converges within only five iterations with the specified parameter setting.

%%%%%%%%%%%%%%%%%%%%%%%%%%%%%%
\begin{figure}[!t]

\hspace{-6 mm}
    \begin{tabular}{cc}
    \subfigure[The transmit power of node 1]
    {
    \begin{minipage}[t]{0.5\textwidth}
    \centering
    {\includegraphics[width = 3.7in] {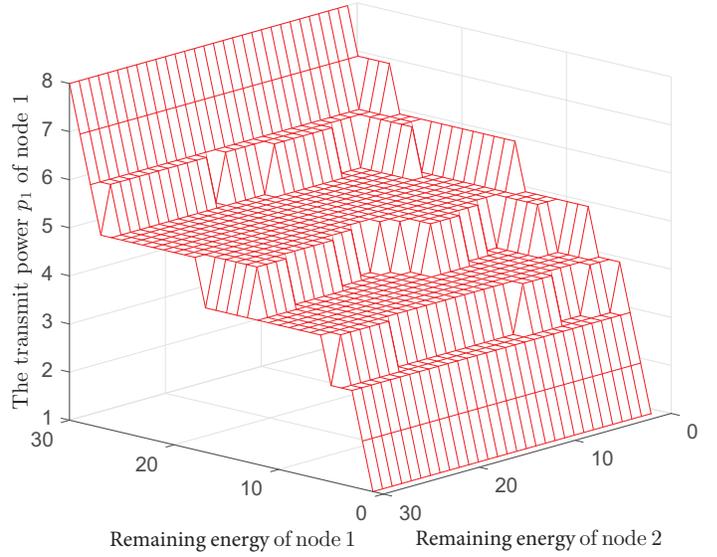} \label{fig:p1}}
    \end{minipage}
    }\\

 \subfigure[The transmit power of node 2]
    {
    \begin{minipage}[t]{0.5\textwidth}
    \centering
    {\includegraphics[width = 3.7in] {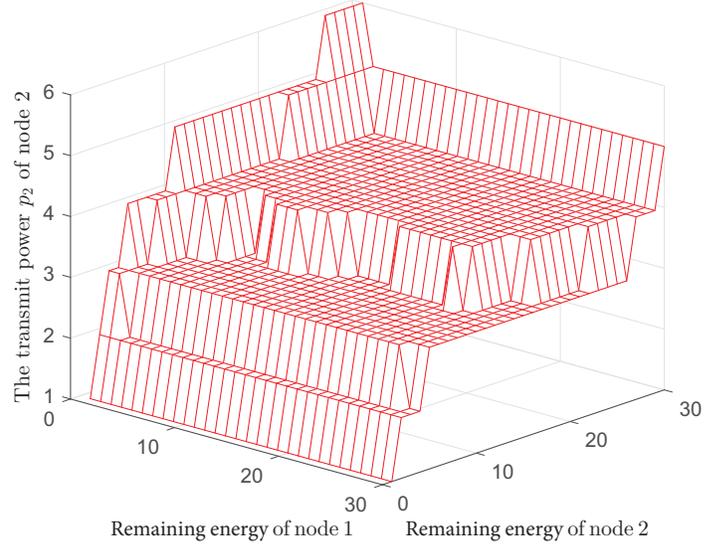} \label{fig:p2}}
    \end{minipage}
    }

    \end{tabular}

\caption{The optimal online power control.} \label{fig:online_power}
\end{figure}
%%%%%%%%%%%%%%%%%%%%%%%%%%%%%%

\subsection{Online Power Control}
For the online case, we set the weighting coefficient of the cost function (see \eqref{df:cost}) to $\alpha=0.99$ and the normalizing constant to $\delta=1$.
    The sizes of the energy buffer of the two nodes are $L_1=30$ and $L_2=30$, respectively.
Thus, we have $L=L_1L_2=900$ and the energy state space is $\Omega=\{1,\cdots, 900\}$.
    We assume that for both nodes, the harvested energy in one slot is a uniformly distributed integer between one and $e_{k,\max}$, where $e_{1,\max}=8$ and $e_{2,\max}=5$.
We also assume that the energy harvesting processes of the two nodes are independent from each other.

For two energy states $s$ and $t$, we denote the corresponding remaining energy pair as $(i,j)$ and $(i', j')$, respectively.
   Under control function $\rho$ and starting from state $s$, we assume that the transmit power of the two nodes is $p_1$ and $p_2$, respectively.
For the given energy states $s, t$ and transmit powers $p_1, p_2$, the uncertainty in transferring from state $s$ to state $t$ is due to the randomness of the energy harvesting process of the two nodes, which is independent from each other.
    Thus, the transfer probability would be $q_{st}=\Pr\{i\rightarrow i'\}\Pr\{j\rightarrow j'\}$.
    Note that $\Pr\{i\rightarrow i'\}=0$ if $i'<i+1-p_1~\text{or}~i'>+e_{1,\max}-p_1$.
Moreover, $\Pr\{i\rightarrow i'\}=\frac{1}{e_{1,\max}}$ if $i'<L$ and $\Pr\{i\rightarrow i'\}=\frac{i+e_{1,\max}-p_1-L}{e_{1,\max}}$ if $i'=L$.
    By performing a similar analysis on $\Pr\{j\rightarrow j'\}$, the transfer probability $q_{st}$ and the transfer matrix $\textbf{P}_\rho$ can be obtained.
%%%\begin{equation}
%%%    \Pr\{i\rightarrow i'\}=
%%%        \left\{
%%%                    \begin{aligned}
%%%                    & 0      & \text{if}~i'<i+1-p_1~\text{or}~i'<>+e_{1,\max}-p_1, \\
%%%                    &   \left\{
%%%                                \begin{aligned}
%%%                                 &\frac{i+e_{1,\max}-p_1-L}{e_{1,\max}}    & \text{if}~ i+e_{1,\max}-p_1>L, \\
%%%                                 &\frac{1}{e_{1,\max}}                                     & \text{otherwise} \\
%%%                                \end{aligned}
%%%                         \right.          & \text{otherwise} \\
%%%                    \end{aligned}
%%%           \right.
%%%\end{equation}

Following Algorithm \ref{alg:online}, we obtain the optimal control function of each node, which specifies the transmit power of each node for each energy state (corresponds to the remaining energy of nodes), as shown in Fig. \ref{fig:p1} and Fig. \ref{fig:p2}.
    We can observe that the transmit power of a node depends mainly on its own remaining energy and is not much affected by the remaining energy of the other node.
In general, $p_k$ is an increasing function of its  remaining energy. However, it is neither convex nor monotonically  increasing with the remaining energy of the other node.
    Note that $p_k$ is obtained by jointly optimizing the cost over $p_k$ and $p_{\tilde{k}}$. Thus, the optimal $p_k$  also has a generalized water-filling interpretation like the offline power control (see Fig. \ref{fig:p12k}).

%%%%%%%%%%%%%%%%%%%%%%%%%%%%%%
\begin{figure*}[htp]

\hspace{-6 mm}
    \begin{tabular}{cc}
    \subfigure[The stationary distribution of energy state]
    {
    \begin{minipage}[t]{0.5\textwidth}
    \centering
    {\includegraphics[width = 3.7in] {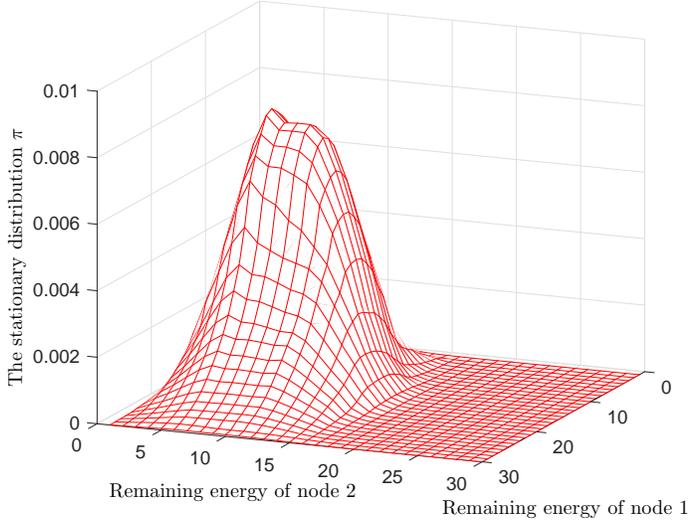} \label{fig:pi_star}}
    \end{minipage}
    }

 \subfigure[The minimum cost versus energy state]
    {
    \begin{minipage}[t]{0.5\textwidth}
    \centering
    {\includegraphics[width = 3.7in] {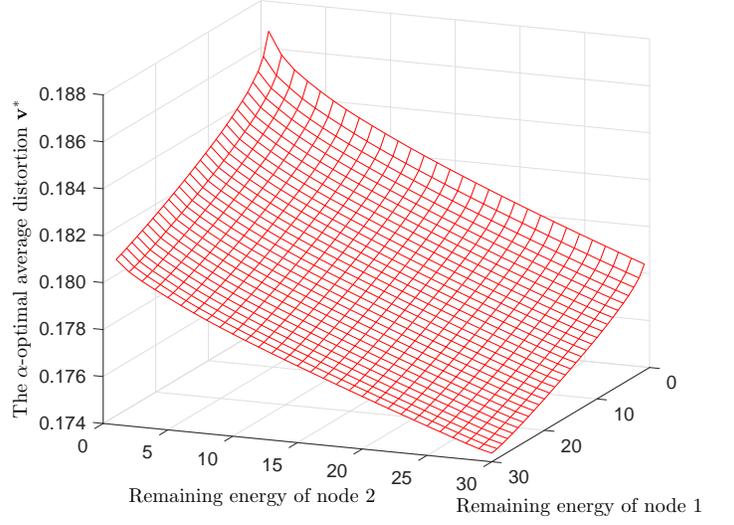} \label{fig:v_star}}
    \end{minipage}
    }\\

    \hspace{0.01\textwidth}
    \subfigure[Convergence of searching $\boldsymbol{v}^*$ iteratively.]
    {
    \begin{minipage}[t]{0.5\textwidth}
    \centering
    {\includegraphics[width = 3.7in] {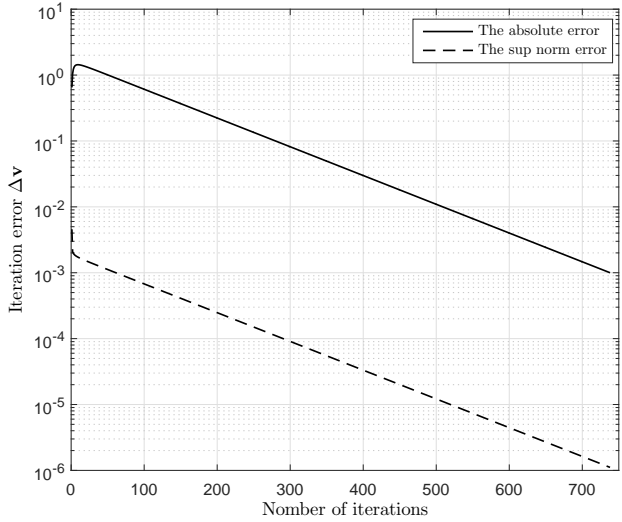} \label{fig:converg}}
    \end{minipage}
    }

    \hspace{0.01\textwidth}
    \subfigure[The minimum expected distortion versus correlation coefficient $\sqrt{\eta}$]
    {
    \begin{minipage}[t]{0.5\textwidth}
    \centering
    {\includegraphics[width = 3.7in] {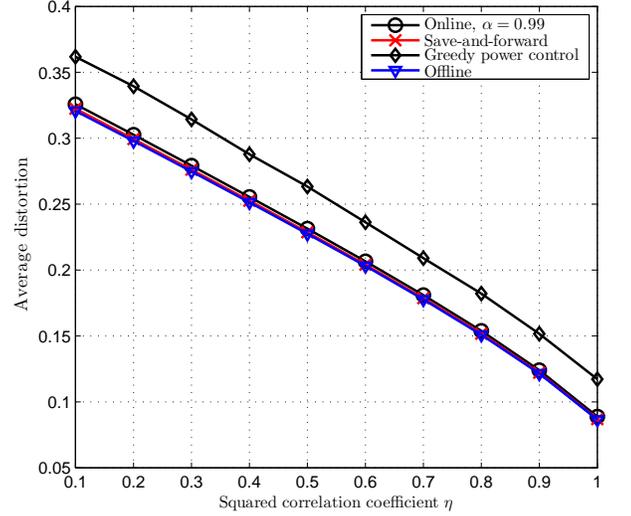} \label{fig:d_rho}}
    \end{minipage}
    }

    \end{tabular}

\caption{On the online power allocation policy.} \label{fig:online_policy}
\end{figure*}
%%%%%%%%%%%%%%%%%%%%%%%%%%%%%%

We plot the stationary distribution of the energy state of the system in Fig. \ref{fig:pi_star}.
    As observed in the figure, the probability that the two nodes have much remaining energy is close to zero.
Therefore, under the optimal power control, the energy buffers of the two nodes are stable.
    This also implies that we do not need very large energy buffers in practical energy harvesting systems.
Fig. \ref{fig:v_star} depicts the minimum achievable cost for different energy states.
    As expected, the cost decreases with the remaining energy of nodes.
However, the cost becomes unaffected by the remaining energy of the two nodes when they are very large.
    This is because, even when the buffers are full, the corresponding transmit powers are not very large, as shown in Fig. \ref{fig:p1}.
%Moreover, it is noted that the probability of having full energy buffers is actually zero, as shown in Fig. \ref{fig:pi_star}.

Fig. \ref{fig:converg} displays the convergence of Algorithm \ref{alg:online}.
    Both the absolute error $\Delta_{|\cdot|} \boldsymbol{v}=\sum_{s=1}^L |\boldsymbol{v}^l(s)-\boldsymbol{v}^{l-1}(s)|$ and the sup norm error $\Delta_{\max} \boldsymbol{v}=\sup_{s} |\boldsymbol{v}^l(s)-\boldsymbol{v}^{l-1}(s)|$ are presented.
 It is seen that the error decreases geometrically, demonstrating the effectiveness of Algorithm \ref{alg:online}.

 We then investigate how the minimum expected distortion changes with the correlation between the two nodes in Fig. \ref{fig:d_rho}.
    In particular, we investigate the performance of the following four schemes:
 1) the online power allocation based on Algorithm~\ref{alg:online};
 2) the offline power allocation based on Algorithm \ref{Alg:iter_wt_fil};
 3) the greedy power allocation where each node uses all the harvested energy in each slot;
 4) the save-and-forward power allocation where each node saves all the harvested energy for a long period of $h(T)=o(T)$ slots and transmits information in the rest of the period \cite{Ulukus-2012-awgn}.
    It is clear that the greedy policy is the most straightforward scheme but it performs the worst.
 On the contrary, the offline policy serves as a strict upper bound of the achievable performance due to the non-causal information about the energy harvesting process.
    Furthermore, the save-and-forward policy has been shown to be the performance limit achieving policy \cite{Ulukus-2012-awgn}.
 As observed in Fig. \ref{fig:d_rho}, our online policy largely outperforms the greedy policy, and performs similar to the offline policy and the save-and-forward policy.

\begin{figure}[!t]
\centering
\includegraphics[width=3.7in]{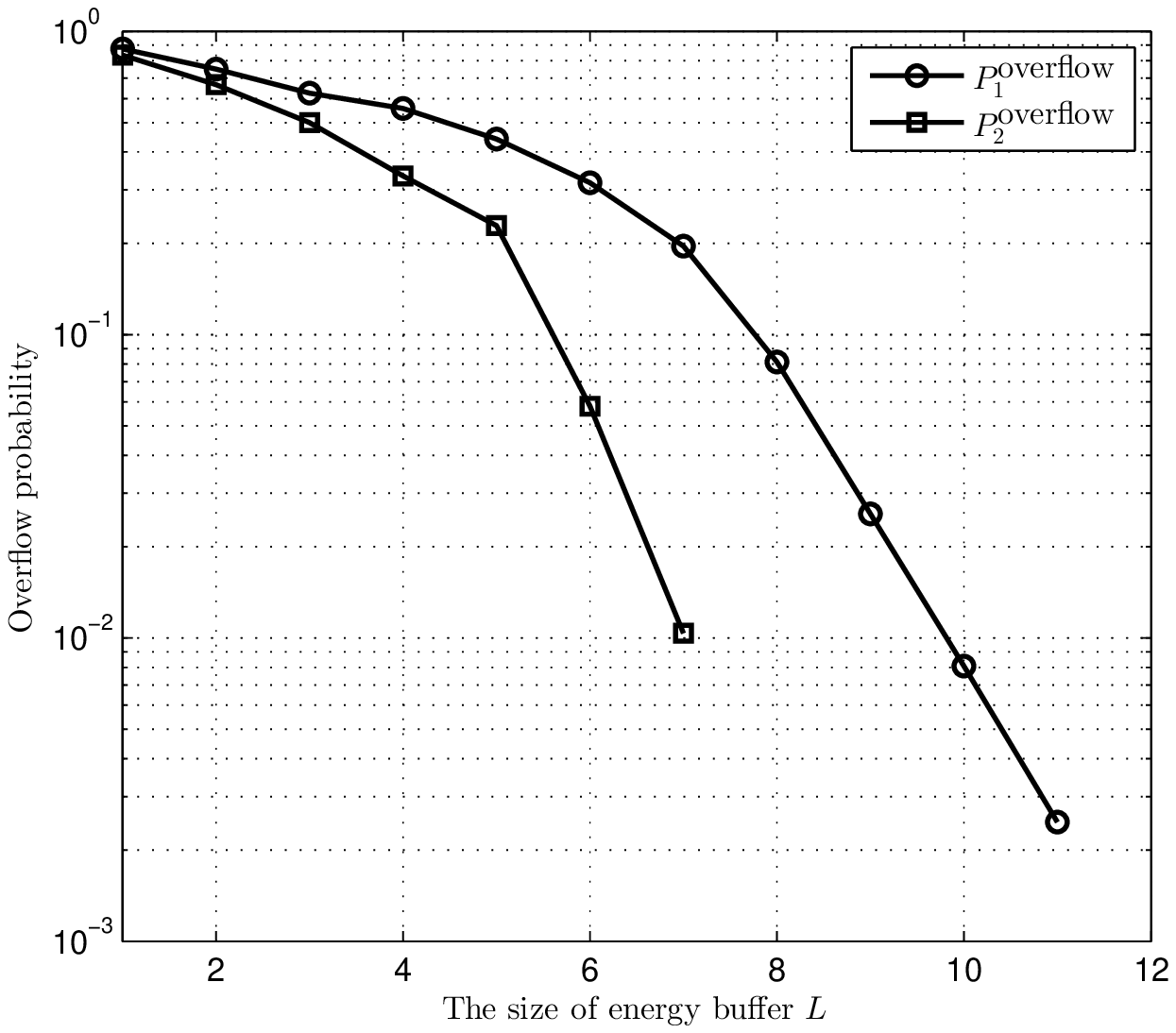}
\caption{The overflow probability of the energy buffer of the two nodes, where $L_1=L_2=L$, $e_{1,\max}=8$, and $e_{2,\max}=6$.} \label{fig:overflow}
\end{figure}

In Fig. \ref{fig:overflow}, we investigate the overflow probability of the energy buffers.
    We set the buffer size to $L_1=L_2$ and the maximum harvested energy in a slot to $e_{1,\max}=8$ and $e_{2,\max}=6$.
Under the optimal online power control obtained using Algorithm \ref{alg:online}, it can be seen that the overflow probability of both energy buffers decreases rapidly and goes to zero when $L>7$ and $L>11$, respectively.
    This is also in line with our result on the stationary distribution of the remaining energy in the energy buffers (see \eqref{rt:pi_stationary} and Fig. \ref{fig:pi_star}), namely, the probability that either of the energy buffers has much energy is nearly zero.
Therefore, we do not need very large energy buffers in real systems and thus our assumption that the energy buffer is large enough is reasonable.

\section{Conclusion}\label{sec6:conclusion}
In this paper, we have studied the optimal offline and online power control policies to minimize the weighted-sum distortion in transmitting correlated sources under energy harvesting constraints.
    We have shown that, while the offline power control outperforms both the online power control and the greedy power control owing to the non-causal information about the energy harvesting process, our online power control performs very close to the offline power control by exploring the statistics of the energy harvesting process.
    In addition, our analytic framework of cost functions for the online power control can also be applied to other networks.
 {We also have observed that when the correlation between the two sources becomes stronger, the sources would be more compressible and thus smaller distortion can be achieved.
    However, the structure of the optimal power control remains unchanged, even for extreme cases such as $\eta=1$.
Moreover, our results have validated the assumption that the energy buffer at each node is large enough so that the probability of energy overflow would be negligible.}
     {To be specific, under the optimal power control, the probability that the energy buffers have much remaining energy is  zero for both  offline and online cases.
Nevertheless, investigating the power control and distortion performance for transmitting correlated sources using very small energy buffers (e.g., unit-sized battery} \cite{Ulukus-unitsize} {) is also a very interesting problem and will be considered in our future work. }

 %     In future work, we may consider the transmission of multi sources under energy harvesting constraints based on multi-terminal source coding theory and random field theory.
%The best system recovery precision can be achieved by leveraging total information rate and total distortion, and incorporating optimal node activation and optimal power control.

\appendices
\renewcommand{\theequation}{\thesection.\arabic{equation}}
\newcounter{mytempthcnt}
\setcounter{mytempthcnt}{\value{theorem}}
\setcounter{theorem}{2}

\section{Proof of Proposition \ref{lem:covex1}} \label{prf:lem1}

\begin{proof}
From Fig. \ref{fig:d1d2}, it is clear that the minimum weighted-sum distortion $D(\boldsymbol{r})$ occurs at some point on  curve segment CD or its two end points (C and D).
    Since it is assumed that $w_1<w_2$, we are focused on curve MD and point D.

Since the coordinate of point D is $\big(\frac{d_{12}^{\min}}{d_2^{\min}},d_2^{\min}\big)$ and curve segment MD is written as  $D_2= \frac{d_{12}^{\min}}{D_1}$, the derivative at point D is given by
\begin{equation*}
  \frac{\textrm{d} D_2}{D_1}=-\frac{d_{12}^{\min}}{D_1^2} = -\frac{(d_{2}^{\min})^2}{d_{12}^{\min}}.
\end{equation*}

    Let $\kappa=-\frac{w_1}{w_2}$ be the slope of line $w_1D_1+w_2D_2=c_0$.
If $\frac{\textrm{d} D_2}{D_1}>\kappa$, then the minimum of $D(\boldsymbol{r})$ occurs at point D, i.e., $D^{\textrm{D}}({\boldsymbol{r}}) = w_1 \frac{d_{12}^{\min}} {d_{2}^{\min}} + w_2 d_2^{\min} $.
    By solving  $\frac{\textrm{d} D_2}{D_1}>\kappa$, we have
\begin{equation*}
  r_2< g(r_1)= -\frac12 \log_2 \frac{w_1\bar{\rho}2^{-2r_1}} {w_2(\bar{\rho}+\rho2^{-2r_1})^2-w_1\rho 2^{-4r_1}}.
\end{equation*}

If $r_2> g(r_1)$, the minimum of  $D(\boldsymbol{r})$  occurs at some point on curve segment MD, where the slope is exactly $\kappa$.
    Solving $D_1$ from $-\frac{d_{12}^{\min}}{D_1^2}=-\frac{w_1}{w_2}$, we have $D_1=\sqrt{\frac{w_2}{w_1}d_{12}^{\min}}$.
Together with $D_1D_2=d_{12}^{\min}$, we finally obtain  $D_2=\sqrt{\frac{w_1}{w_2}d_{12}^{\min}}$ and $D^{\textrm{MD}}({\boldsymbol{r}}) = 2\sqrt{w_1w_2 d_{12}^{\min}}$.
    Thus,
    \begin{equation*}\label{rt:D_r_apx}
      D(\boldsymbol{r}) =\left\{
            \begin{aligned}
                    &D^{\textrm{MD}}({\boldsymbol{r}})  & \textrm{if}~ r_2\geq g(r_1),\\
                    & D^{\textrm{D}}({\boldsymbol{r}})  & \textrm{if}~r_2< g(r_1).
            \end{aligned}
            \right.
    \end{equation*}

To prove the convexity of $D(\boldsymbol{r})$, we first investigate the difference between $D^{\textrm{D}}({\boldsymbol{r}})$ and $D^{\textrm{MD}}({\boldsymbol{r}}) $,
    \begin{equation*}
    \begin{split}
        D^{\textrm{D}}({\boldsymbol{r}})-D^{\textrm{MD}}({\boldsymbol{r}}) &= w_1 \frac{d_{12}^{\min}}{d_{2}^{\min}} + w_2 d_2^{\min} - 2\sqrt{w_1w_2 d_{12}^{\min}} \\
        & = \left(\sqrt{w_1 \frac{d_{12}^{\min}}{d_{2}^{\min}}}- \sqrt{w_2 d_2^{\min}} \right)^2 \geq 0,
      \end{split}
    \end{equation*}
   where the equality holds if $r_2=g(r_1)$.

This means that the surface of $D^{\textrm{D}}({\boldsymbol{r}})$ intersects the surface of $D^{\textrm{MD}}({\boldsymbol{r}}) $ only on one curve.
    By evaluating their first and second order derivatives, one can show that both $D^{\textrm{D}}({\boldsymbol{r}})$ and $D^{\textrm{MD}}({\boldsymbol{r}}) $ are decreasing and convex in ${\boldsymbol{r}}$.
Therefore, the surface of $D^{\textrm{D}}({\boldsymbol{r}})$ is tangent with that of $D^{\textrm{MD}}({\boldsymbol{r}}) $, which implies
$D({\boldsymbol{r}})$ is also decreasing and convex in coding rate $\boldsymbol{r}$.
        Since $r_k=\frac12 \log_2(1+h_kp_k)$ is concave in $p_k$,
    we know that $D(\boldsymbol{p})$ is convex in $\boldsymbol{p}$ \cite{cv_byod-2004}.
\end{proof}

\section{Proof of Proposition \ref{prop:lambda0} and \ref{prop:lambda2}}
\subsection{Proof of Proposition \ref{prop:lambda0}} \label{prf:prop1}
\begin{proof}
The complimentary slackness conditions associated with $(\mathrm{P}_3)$ are as follows,
 \begin{eqnarray}    % for double column
    \label{slk:lambda1}
\lambda_{1\tau}( -D_{1\tau}+(\bar{\eta}+\eta2^{-2r_{2\tau}})2^{-2r_{1\tau}} ) \hspace{-2mm}&=&\hspace{-2mm}0,\\
    \label{slk:lambda2}
\lambda_{2\tau}(-D_{2\tau}+(\bar{\eta}+\eta2^{-2r_{1\tau}})2^{-2r_{2\tau}} )\hspace{-2mm}&=&\hspace{-2mm}0,\\
    \nonumber
\lambda_{3\tau}(-\log_2 D_{1\tau}-\log_2 D_{2\tau} - 2(r_{1\tau}+r_{2\tau}) \hspace{-2mm}& &\hspace{-2mm}\\
    \label{slk:lambda3}
+ \log_2(\bar{\eta}+\eta2^{-2(r_{1\tau}+r_{2\tau})}) )\hspace{-2mm}&=&\hspace{-2mm}0,\\
    \label{slk:lambda4}
\lambda_{4\tau}((D_{1\tau}-\bar{\eta})2^{2r_{2\tau}} -\eta)\hspace{-2mm}&=&\hspace{-2mm}0, \\
    \label{slk:lambda5}
\lambda_{5\tau}( (D_{2\tau}-\bar{\eta})2^{2r_{1\tau}} -\eta) \hspace{-2mm}&=&\hspace{-2mm}0, \\
    \label{slack:cstr_mu}
\mu_{k\tau}\sum_{i=1}^\tau \Big( \frac{1}{h_k}(2^{2r_{ki}}-1) - e_{ki}\Big) \hspace{-2mm}&=&\hspace{-2mm}0,\\
    \label{slack:cstr_theta}
\theta_{k\tau}r_{k\tau} \hspace{-2mm}&=&\hspace{-2mm}0,  \\ \nonumber
\forall~k=1,2,~~~ \tau=1,\cdots,T. \hspace{-3mm}& &\hspace{-2mm}
\end{eqnarray}
 %\begin{eqnarray}
%    \label{slk:lambda1}
%\lambda_{1\tau}( -D_{1\tau}+(\bar{\eta}+\eta2^{-2r_{2\tau}})2^{-2r_{1\tau}} ) \hspace{-2mm}&=&\hspace{-2mm}0,\\
%    \label{slk:lambda2}
%\lambda_{2\tau}(-D_{2\tau}+(\bar{\eta}+\eta2^{-2r_{1\tau}})2^{-2r_{2\tau}} )\hspace{-2mm}&=&\hspace{-2mm}0,\\
%\label{slk:lambda3}
%\lambda_{3\tau}(-\log_2 D_{1\tau}-\log_2 D_{2\tau} - 2(r_{1\tau}+r_{2\tau})
%+ \log_2(\bar{\eta}+\eta2^{-2(r_{1\tau}+r_{2\tau})}) )\hspace{-2mm}&=&\hspace{-2mm}0,\\
%    \label{slk:lambda4}
%\lambda_{4\tau}((D_{1\tau}-\bar{\eta})2^{2r_{2\tau}} -\eta)\hspace{-2mm}&=&\hspace{-2mm}0, \\
%    \label{slk:lambda5}
%\lambda_{5\tau}( (D_{2\tau}-\bar{\eta})2^{2r_{1\tau}} -\eta) \hspace{-2mm}&=&\hspace{-2mm}0, \\
%    \label{slack:cstr_mu}
%\mu_{k\tau}\sum_{i=1}^\tau \Big( \frac{1}{h_k}(2^{2r_{ki}}-1) - e_{ki}\Big) \hspace{-2mm}&=&\hspace{-2mm}0,\\
%    \label{slack:cstr_theta}
%\theta_{k\tau}r_{k\tau} \hspace{-2mm}&=&\hspace{-2mm}0,  \\ \nonumber
%\forall~k=1,2,~~~ \tau=1,\cdots,T. \hspace{-3mm}& &\hspace{-2mm}
%\end{eqnarray}

    Note that for any given coding rate pair  $(r_{1\tau}, r_{2\tau})$, the minimum weighted-sum distortion $D(\boldsymbol{r})$ occurs at some point on curve segment CD or the two end points (C or D).
        Since it is assumed $w_1<w_2$, we can focus on curve MD and point D.
    Therefore, constraint \eqref{cst:lambda1}, \eqref{cst:lambda4}, and \eqref{cst:lambda5} are never active, while constraint \eqref{cst:lambda3} is always active.
        Using this together with the complementary slackness conditions \eqref{slk:lambda1}, \eqref{slk:lambda3}--\eqref{slk:lambda5}, the proposition is proved.
\end{proof}

\subsection{Proof of Proposition \ref{prop:lambda2}} \label{prf:prop2}
\begin{proof}
Following the same analysis in Appendix \ref{prf:lem1}, the minimum weighted-sum distortion occurs at point $\text{D}$ if $\frac{\textrm{d} D_2}{D_1}<-\frac{w_1}{w_2}$.
In this case, constraint (\ref{cst:lambda2}) is active, which implies $\lambda_{2\tau}>0$.
    By solving $r_2$ from $\frac{\textrm{d} D_2}{D_1}<-\frac{w_1}{w_2}$, we have
\begin{equation*}
  r_{2\tau}< g(r_{1\tau})= -\frac12 \log_2 \frac{w_1\bar{\eta}2^{-2r_{1\tau}}} {w_2(\bar{\eta}+\eta2^{-2r_{1\tau}})^2-w_1\eta 2^{-4r_{1\tau}}}.
\end{equation*}

On the other hand, if $r_2\geq g(r_1)$ is true, the minimum sum distortion occurs at some point on curve segment MD, where the derivative is exactly $-\frac{w_1}{w_2}$.
    Therefore, the constraint (\ref{cst:lambda2}) is not active and we have $\lambda_{2\tau}>0$.
This completes the proof.
\end{proof}

\section{Proof of Theorem \ref{th:p_structure}} \label{prf:thm1}
\begin{proof}
According to  slackness condition (\ref{slack:cstr_mu}), we know that $\mu_{k\tau}>0$ holds  if the energy buffer is emptied (i.e., $\sum_{i=1}^\tau(p_{ki}-e_{ki})=0$) and
    $\mu_{k\tau}=0$ otherwise.
It can be readily seen that $\nu_{k\tau}=1/\sum_{i=\tau}^T \mu_{ki}$ is monotonically increasing with time $\tau$ and does not change until the energy buffer is depleted.
    Also note that a node will not deplete its energy buffer unless it harvests more energy on average in the following slots than in previous slots (otherwise, some energy should be saved in the buffer for the following slots).

    To prove the second part of the theorem, we present the first order derivative of  the Lagrangian \eqref{eq:lagrange} as follows.
In particular, we have $\lambda_{1\tau}=0, \lambda_{4\tau}=0$, and $\lambda_{5\tau}=0$ for all $k$ and $\tau$ by Proposition \ref{prop:lambda0}--\ref{prop:lambda2}.
    %Differentiating the Lagrangian  with respect to $D_{k\tau}$ and $r_{k\tau} $ respectively, we have
\begin{eqnarray}
\label{eq:lag_dif_d1}
   \frac{\partial \mathcal{L}}{\partial D_{1\tau}}\hspace{-2.75mm} &=& \hspace{-2.75mm} w_1 - \frac{\lambda_{3\tau}} {D_{1\tau}\ln 2}, \\
\label{eq:lag_dif_d2}
  \frac{\partial \mathcal{L}}{\partial D_{2\tau}}\hspace{-2.75mm} &=& \hspace{-2.75mm} w_2-\lambda_{2\tau} - \frac{\lambda_{3\tau}} {D_{2\tau}\ln 2}, \\
\label{eq:lag_dif_r1} \nonumber
   \frac{\partial \mathcal{L}}{\partial r_{1\tau}} \hspace{-2.75mm} &=& \hspace{-2.75mm}
          - \lambda_{2\tau}\eta 2^{-2(r_{1\tau}+r_{2\tau})} - \frac{\lambda_{3\tau}\eta2^{-2(r_{1\tau}+r_{2\tau})}} {\bar{\eta}+\eta 2^{-2(r_{1\tau}+r_{2\tau})}}  \\
%\label{eq:lag_dif_r1}
    \hspace{-2.75mm} & & \hspace{-2.75mm}
                        -\frac{\lambda_{3\tau}}{\ln 2} + \frac{2^{2r_{1\tau}}}{h_1\nu_{1\tau}}  - \frac{\theta_{1\tau}}{2\ln 2},\\
\label{eq:lag_dif_r2}  \nonumber
   \frac{\partial \mathcal{L}}{\partial r_{2\tau}}\hspace{-2.75mm} &=& \hspace{-2.75mm}
           - \lambda_{2\tau}\bar{\eta} 2^{-2r_{2\tau}} \hspace{-.3mm} - \hspace{-.3mm} \lambda_{2\tau}\eta 2^{-2(r_{1\tau}+r_{2\tau})}
            \hspace{-.3mm} -\hspace{-.3mm}\frac{\lambda_{3\tau}}{\ln 2} \hspace{-.3mm}  \\
\label{eq:lag_dif_r2}
    \hspace{-2.75mm} & & \hspace{-2.75mm}
            - \hspace{-.3mm} \frac{\lambda_{3\tau}\eta2^{-2(r_{1\tau}+r_{2\tau})}}
            {\bar{\eta}+\eta 2^{-2(r_{1\tau}+r_{2\tau})}} \hspace{-.3mm} +\hspace{-.3mm}
          \frac{2^{2r_{2\tau}}}{h_2\nu_{2\tau}}  \hspace{-.3mm} - \hspace{-.3mm} \frac{\theta_{2\tau}}{2\ln 2}.
\end{eqnarray}

We first consider the case of $\lambda_{2\tau}=0$, i.e., the minimum distortion occurs on curve segment MD.
In this case, the distortion of the nodes is given by $D_1=\sqrt{\frac{w_2}{w_1}d_{12}^{\min}}$ and $D_2=\sqrt{\frac{w_1}{w_2}d_{12}^{\min}}$, respectively.
    By setting the derivatives (\ref{eq:lag_dif_d1}) to zero and solving for $\lambda_{3\tau}$, we have $\lambda_{3\tau}=\ln 2\sqrt{w_1w_2(\bar{\eta}+\eta xy)xy}$.
By setting the derivatives in (\ref{eq:lag_dif_r1}) and (\ref{eq:lag_dif_r2})  to be zero, we have
    \begin{eqnarray*}
      F_{11}(x,y,\nu_{1\tau},\nu_{2\tau}) \hspace{-2.75mm} &=& \hspace{-2.75mm} \frac{\lambda_{3\tau}\eta xy}{\bar{\eta}+\eta xy}
                    + \frac{\lambda_{3\tau}}{\ln 2} - \frac{1}{h_1\nu_{1\tau}}\frac{1}{x} + \frac{\theta_{1\tau}}{2\ln 2}=0,\\
      F_{12}(x,y,\nu_{1\tau},\nu_{2\tau})  \hspace{-2.75mm} &=& \hspace{-2.75mm} \frac{1}{h_2\nu_{2\tau}}\frac{1}{y} -\frac{1}{h_1\nu_{1\tau}}\frac{1}{x} + \frac{\theta_{1\tau}-\theta_{2\tau}}{2\ln 2}=0.
    \end{eqnarray*}
Taking the derivative $F_{11}$ with respect to $x$ and $\nu_{1\tau}$, we have
\begin{eqnarray*}
  \frac{\partial F_{11}}{\partial x} \hspace{-2.75mm} &=& \hspace{-2.75mm} \ln2 \sqrt{w_1w_2} \frac{ 3\bar{\eta}\sqrt{xy}+3\eta(xy)^{\frac32} -\eta x^{\frac32}y^{\frac52} }{2(\bar{\eta}+\eta xy)^{\frac32}} \\
  \hspace{-2.75mm} && \hspace{-2.75mm}
  +\frac{(\bar{\eta}+2\eta xy)y}{2\sqrt{(\bar{\eta}+\eta xy)xy}} +\frac{1}{h_1\nu_{1\tau}}\frac{1}{x^2} >0,\\
  \frac{\partial F_{11}}{\partial \nu_{1\tau}}
  \hspace{-2.75mm} &=& \hspace{-2.75mm}  \frac{1}{h_1x\nu_{1\tau}^2}>0.
\end{eqnarray*}
Hence,
\begin{equation*}
    \frac{\partial x}{\partial\nu_{1\tau}} = -\frac{\frac{\partial F_{11}}{\partial x} }{\frac{\partial F_{11}}{\partial \nu_{1\tau}}}  <0,
\end{equation*}
which implies that $x$ is decreasing with $\nu_{1\tau}$.

Likewise, we have
\begin{equation*}
\begin{split}
   \frac{\partial F_{22}}{\partial x} =& \frac{1}{h_1 x^2\nu_{1\tau}},
   \frac{\partial F_{22}}{\partial y} = \frac{-1}{h_2 y^2\nu_{2\tau}},\\
   \frac{\partial F_{22}}{\partial \nu_{2\tau}} =& \frac{-1}{h_2y\nu_{2\tau}^2},
   \frac{\partial F_{22}}{\partial \nu_{1\tau}} = \frac{1}{h_1x\nu_{1\tau}^2},
   \end{split}
\end{equation*}
and
\begin{equation*}
  \frac{\partial y}{\partial\nu_{2\tau}}<0, \frac{\partial x}{\partial\nu_{2\tau}}>0,   \frac{\partial y}{\partial\nu_{1\tau}}>0.
\end{equation*}
Therefore, we know that $y$ is decreasing with $\nu_{2\tau}$ and increasing with  $\nu_{1\tau}$ while $x$ is increasing with  $\nu_{2\tau}$.

Second, for the case $\lambda_{2\tau}>0$, i.e., the minimum distortion occurs at point D, we have $D_{1\tau}=\frac{d_{12}^{\min}} {d_2^{\min}}$ and $D_{2\tau}=d_2^{\min}$.
   Setting the partial derivatives in  (\ref{eq:lag_dif_d1}) and (\ref{eq:lag_dif_d2}) to zero yields
\begin{eqnarray*}
  \lambda_{3\tau} \hspace{-2.75mm} &=& \hspace{-2.75mm} w_1 \ln2 \frac{(\bar{\eta}+\eta xy)x}{\bar{\eta}+\eta x}, \\
  \lambda_{2\tau} \hspace{-2.75mm} &=& \hspace{-2.75mm} w_2 -  \frac{w_1(\bar{\eta}+\eta xy)x}{(\bar{\eta}+\eta x)^2y}.
\end{eqnarray*}

Following a similar argument, we have
    \begin{eqnarray*}
      F_{21}(x,y,\nu_{1\tau},\nu_{2\tau}) \hspace{-3.25mm} &=& \hspace{-3.25mm} \lambda_{2\tau}\eta xy
                       + \frac{\lambda_{3\tau}\eta xy}{\bar{\eta}+\eta xy}
                       + \frac{\lambda_{3\tau}}{\ln 2}  \\
                    \hspace{-3.25mm} && \hspace{-3.25mm}
                    - \frac{1}{h_1\nu_{1\tau}}\frac{1}{x} + \frac{\theta_{1\tau}}{2\ln 2} =0,\\
      F_{22}(x,y,\nu_{1\tau},\nu_{2\tau})  \hspace{-3.25mm} &=& \hspace{-3.25mm}
                       - \lambda_{2\tau}\eta y  +  \frac{\theta_{1\tau} - \theta_{2\tau}}{2\ln 2}  \\
                    \hspace{-3.25mm} && \hspace{-3.25mm}
                    - \frac{1}{h_2\nu_{2\tau}}\frac{1}{y}  - \frac{1}{h_1\nu_{1\tau}}\frac{1}{x}
                     =  0.
    \end{eqnarray*}
Hence,
\begin{equation*}
  \frac{\partial x}{\partial\nu_{1\tau}}<0,~~ \frac{\partial y}{\partial\nu_{2\tau}}<0,~~ \frac{\partial x}{\partial\nu_{2\tau}}>0, ~~  \frac{\partial y}{\partial\nu_{1\tau}}>0,
\end{equation*}
i.e., $x$ is increasing with  $\nu_{2\tau}$ and decreasing with $\nu_{1\tau}$, while $y$ is decreasing with $\nu_{2\tau}$ and increasing with  $\nu_{1\tau}$.

Note that $p_{1\tau}=\frac{1}{h_1}(\frac{1}{x}-1)$ and $p_{2\tau}=\frac{1}{h_2}(\frac{1}{y}-1)$ are decreasing with respect to $x$ and $y$, respectively.
    Thus, we know that $p_{k\tau}$ is increasing with $\nu_{k\tau}$ and decreasing with $\nu_{\tilde{k}\tau}$.
Moreover, $\nu_{k\tau}$ would not be changed unless the energy buffer of node $k$ is depleted.
    This implies that $p_{k\tau}$ should be increased if its energy buffer is depleted and should be decreased if the energy buffer of the other node is depleted.
This completes the proof.
\end{proof}

\section{Proof of Theorem \ref{th:3cost_equal}} \label{prf:3cost_equal}

\begin{proof}
    For any policy $\varrho$, the following equality holds true:
    \begin{equation}\label{eq:avg_cdt0}
            \mathbb{E}\left( \sum_{\tau=1}^{T} \big( \alpha v_{\varrho_\tau}(S_\tau)  -
                                           \alpha \mathbb{E}(v_{\varrho_\tau}(S_\tau)|S_{\tau-1})   \big) \right) =0,
    \end{equation}
    which follows from the equality $\mathbb{E}(\mathbb{E}(X|Y))=\mathbb{E}(X)$.

    For any $\tau\geq1$, we further have
    \begin{eqnarray}
        \nonumber
           \hspace{-3mm} &&\hspace{-3mm}   \alpha \mathbb{E}(v_{\varrho_\tau}(S_\tau)|S_{\tau-1}=s) = \sum_{t=1}^{L}
                                                                        \alpha q_{st} v_{\varrho_{\tau}}(t) \\
        \nonumber
          \hspace{-3mm} &&\hspace{-3mm}  =\bar\alpha d_{\rho_{\tau-1}}(s) + \sum_{t=1}^{L}
                                                                        q_{st} v_{\varrho_{\tau}}(t)  - \bar\alpha {d}_{\rho_{\tau-1}}(s)\\
        \label{dr:prop3_1}
          \hspace{-3mm} &&\hspace{-3mm}  =v_{\varrho_{\tau-1}}(s)  - \bar\alpha d_{\rho_{\tau-1}}(s) ,
    \end{eqnarray}
    where $q_{st}$ is the transfer probability from state $s$ to state $t$, i.e., the $t$-th element in the $s$-th row of $\textbf{P}_{\rho_{\tau-1}}$.

    Substituting $\alpha \mathbb{E}(v_{\varrho_\tau}(X_\tau)|X_{\tau-1}=s)$ in (\ref{eq:avg_cdt0}) with (\ref{dr:prop3_1}), we have
\begin{eqnarray} \label{dr:prop1_2}
\nonumber
\hspace{-9mm} &&\hspace{-3mm} \sum_{\tau=1}^{T} \big( \alpha v_{\varrho_\tau}(S_\tau)  -
                                           \alpha \mathbb{E}(v_{\varrho_\tau}(S_\tau)|S_{\tau-1})   \big) \\
\nonumber
\hspace{-9mm} &&\hspace{-3mm} =\sum_{\tau=1}^{T} \big( \alpha v_{\varrho_\tau}(S_\tau)
                                            - v_{\varrho_{\tau-1}}(S_{\tau-1})  + \bar\alpha d_{\rho_{\tau-1}}(S_{\tau-1}) \big) \\
\nonumber
\hspace{-9mm} &&\hspace{-3mm} =-\sum_{\tau=1}^{T-1} \bar\alpha \big(v_{\varrho_\tau}(S_\tau) -  d_{\rho_{\tau-1}}(S_{\tau-1}) \big) \\
\label{dr:prop3_2}
\hspace{-9mm} &&\hspace{-3mm} ~~~~  + \alpha v_{\varrho_{T}}(S_{T}) - v_{\varrho_{0}}(S_{0}) + \bar\alpha d_{\rho_{0}}(S_{0}).
\end{eqnarray}

Under the same  control function $\rho_0$, we have $\mathbb{E}(v_{\varrho_{0}}(S_{0})) = \mathbb{E}(v_{\varrho_{0}}(S_{T}))$ and $\mathbb{E}(d_{\rho_{0}}(S_{0})) = \mathbb{E}(d_{\rho_{T}}(S_{T}))$.
    Thus, by taking the expectation on both sides of \eqref{dr:prop3_2} and applying \eqref{eq:avg_cdt0}, we have
\begin{eqnarray}
\nonumber
0 \hspace{-3mm} &=&\hspace{-3mm} \mathbb{E}\left( \sum_{\tau=1}^{T} \big( \alpha v_{\varrho_\tau}(S_\tau)  -
                                           \alpha \mathbb{E}(v_{\varrho_\tau}(S_\tau)|S_{\tau-1})   \big) \right)\\
\label{dr:prop3_3}
  \hspace{-3mm} &=&\hspace{-3mm} \sum_{\tau=1}^{T} \bar\alpha \big(v_{\varrho_\tau}(S_\tau) -  d_{\rho_{\tau-1}}(S_{\tau-1}) \big).
\end{eqnarray}

Theorem \ref{th:3cost_equal} is hence proved by dividing $T$ on both sides of \eqref{dr:prop3_3} and letting $T$ go to infinity.
\end{proof}

\section{Proof  of Theorem \ref{prop:contract4} } \label{prf:contract4}
\begin{proof}
    \begin{eqnarray*}
        \mathbb{T}(\boldsymbol{u}) \hspace{-3mm} &=&\hspace{-3mm}  \min_{\rho}\left\{\bar\alpha \boldsymbol{d}_\rho + \alpha \textbf{P}_\rho \boldsymbol{u}\right\}\\
        \hspace{-3mm} &=&\hspace{-3mm}  \min_{\rho}\left\{\bar\alpha \boldsymbol{d}_\rho + \alpha \textbf{P}_\rho \boldsymbol{v} + \alpha \textbf{P}_\rho (\boldsymbol{u}-\boldsymbol{v})\right\}\\
        \hspace{-3mm} &\leq&\hspace{-3mm}  \min_{\rho}\left\{\bar\alpha \boldsymbol{d}_\rho + \alpha \textbf{P}_\rho \boldsymbol{v} + \alpha \textbf{P}_\rho \Vert\boldsymbol{u}-\boldsymbol{v}\Vert_\infty \textbf{1}\right\}\\
        \hspace{-3mm} &\leq&\hspace{-3mm}   \mathbb{T}(\boldsymbol{v})  + \alpha \Vert\boldsymbol{u}-\boldsymbol{v}\Vert_\infty \textbf{1},
    \end{eqnarray*}
    where $ \textbf{1}$ is an $L$-dimensional vector of ones.

    Likewise, we can also show $\mathbb{T}(\boldsymbol{v}) \leq  \mathbb{T}(\boldsymbol{u})  + \alpha \Vert\boldsymbol{u}-\boldsymbol{v}\Vert_\infty \textbf{1}$, which means
     $\Vert  \mathbb{T}(\boldsymbol{u})-  \mathbb{T}(\boldsymbol{v})\Vert_\infty \leq\alpha \Vert \boldsymbol{u}-\boldsymbol{v}\Vert_\infty$ and thus proves the theorem.
\end{proof}

\section{Proof of Theorem \ref{th:main4}}   \label{prf:main4}

\begin{proof}
To facilitate the proof, for a given control policy $\rho$, we define an updating rule from $\mathbb{R}_{++}^{L}$ to $\mathbb{R}_{++}^{L}$:
\begin{equation*}\label{df:thm_T}
    \mathbb{T}_\rho(\boldsymbol{v}) = \bar\alpha \boldsymbol{d}_\rho + \alpha \textbf{P}_\rho \boldsymbol{v}.
\end{equation*}

Given a power control policy $\varrho=\{\rho_1,\cdots,\rho_T\}$ and an positive initial cost vector $\boldsymbol{v}_0$, by updating $\boldsymbol{v}_0$ with $\rho_\tau (\tau=1,\cdots,T)$ sequentially, we have the following property of $\mathbb{T}$
\begin{eqnarray}
\nonumber
    \hspace{-8mm}&&\hspace{-3mm} \lim_{T\rightarrow\infty} \mathbb{T}_{\rho_1}\mathbb{T}_{\rho_2}\cdots \mathbb{T}_{\rho_T} (\boldsymbol{v}_{0})  \\
\nonumber
    \hspace{-8mm}&&\hspace{-3mm} =\lim_{T\rightarrow\infty} \mathbb{T}_{\rho_2}\cdots \mathbb{T}_{\rho_T}
                                                                    (\bar\alpha \boldsymbol{d}_{\rho_1} +\alpha \textbf{P}_{\rho_1} \boldsymbol{v}_{0} ) \\
    \nonumber
    \hspace{-8mm} && \hspace{-3mm}=\lim_{T\rightarrow\infty} \sum_{\tau=1}^T \bar\alpha \alpha^\tau
                                                                          \textbf{P}_{\rho_1} \cdots\textbf{P}_{\rho_{\tau-1}}\boldsymbol{d}_{\rho_{\tau}}
                                                                            +\alpha^T \textbf{P}_{\rho_{T+1}}\boldsymbol{v}_{0} \\
\label{eq:eqal_vpi}
    \hspace{-8mm} && \hspace{-3mm}= \boldsymbol{v}_{\varrho},
\end{eqnarray}
where (\ref{eq:eqal_vpi}) follows from the definition of $\boldsymbol{v}_{\varrho}$ (see (\ref{df:v_pi_0})) and the fact $0\leq\alpha\leq1$.

Denote the fixed point of $\mathbb{T}(\boldsymbol{v})$ as $\tilde{\boldsymbol{v}}$, we will prove $\tilde{\boldsymbol{v}} = \boldsymbol{v}^*$ in the following subsections.
    Fist, we will prove $\boldsymbol{v}^*\leq \tilde{\boldsymbol{v}}$.

\subsection{ $\boldsymbol{v}^*\leq \tilde{\boldsymbol{v}}$ }

Given $\tilde{\boldsymbol{v}}$, we can find the control function $\rho$ minimizing $\mathbb{T}(\tilde{\boldsymbol{v}})$ by solving (\ref{eq:T_explains}).
    Denote $\varrho'=\{\rho,\rho,\cdots\}$ as a stationary power control policy.
Start from an initial cost vector  $\boldsymbol{v}_0$, we  apply control policy $\varrho$ to  $\boldsymbol{v}_0$ (equivalent to apply $\mathbb{T}(\boldsymbol{v})$) for infinite times.
    According to \eqref{eq:eqal_vpi}, we have
    \begin{equation*}
        \boldsymbol{v}_{\varrho'} = \lim_{T\rightarrow\infty} \mathbb{T}^T (\boldsymbol{v}_{0}).
    \end{equation*}

Since $\mathbb{T}(\boldsymbol{v})$ is a contraction mapping, we know that $\lim_{T\rightarrow\infty} \mathbb{T}^T (\boldsymbol{v}_{0})$ converges to its corresponding fixed point with geometric speed.
    Thus, we have $\boldsymbol{v}_{\varrho'}=\tilde{\boldsymbol{v}}$.

By the definition of $\boldsymbol{v}^*$ (see (\ref{df:v_opt})), we have
\begin{equation}\label{dr:them1_1}
    \boldsymbol{v}^*= \inf_{\varrho} \boldsymbol{v}_{\varrho} \leq \boldsymbol{v}_{\varrho'}=\tilde{\boldsymbol{v}}.
\end{equation}

\subsection{ $\boldsymbol{v}^*\geq \tilde{\boldsymbol{v}}$ }

Let $\varrho=\{\rho_1,\rho_2,\cdots\}$ be the optimal policy achieving $\boldsymbol{v}^*$.
    By the definition of $\mathbb{T}(\boldsymbol{v})$ (see (\ref{df:T})), the following inequalities hold true for any positive $\boldsymbol{v}_0$,
\begin{eqnarray}\label{eq:TandT_rho}
\nonumber
    \mathbb{T}(\boldsymbol{v}_{0}) \hspace{-3mm} &\leq& \hspace{-3mm}  \mathbb{T}_{\rho_1}(\boldsymbol{v}_{0}), \\
\nonumber
    \mathbb{T}^2(\boldsymbol{v}_{0}) \hspace{-3mm} &\leq& \hspace{-3mm} \mathbb{T}(\mathbb{T}_{\rho_1}(\boldsymbol{v}_{0})) \leq \mathbb{T}_{\rho_2}\mathbb{T}_{\rho_1}(\boldsymbol{v}_{0}), \\
    \label{dr:vtildgeqvstar}
    \mathbb{T}^T(\boldsymbol{v}_{0}) \hspace{-3mm} &\leq &\hspace{-3mm} \mathbb{T}_{\rho_T}\cdots \mathbb{T}_{\rho_1}(\boldsymbol{v}_{0}).
\end{eqnarray}

As $T$ goes to infinity, the left-hand side and the right-hand side of (\ref{dr:vtildgeqvstar}) reduce to $\tilde{\boldsymbol{v}}$ and $\boldsymbol{v}^*$, respectively.
    Thus, we have
\begin{equation} \label{dr:them1_2}
  \tilde{\boldsymbol{v}}\leq\boldsymbol{v}^*.
\end{equation}

By combining \eqref{dr:them1_1} and \eqref{dr:them1_2}, we have
\begin{equation*}
  \tilde{\boldsymbol{v}}=\boldsymbol{v}^*.
\end{equation*}
That is, $\boldsymbol{v}^*$ is the fixed point of $\mathbb{T}(\boldsymbol{v})$.
    This completes the proof of the theorem.
\end{proof}

%\section*{Acknowledgement}
%This work was partly supported by the Brain Korea 21 Plus Project in 2015, the ICT R\&D program of MSIP\/IITP [B0126-15-1017, Spectrum Sensing and Future Radio Communication Platforms] and the National Research Foundation of Korea (NRF) grant funded by the Korean government(MSIP) (2014R1A5A1011478).

\small{
\bibliographystyle{IEEEtran}

\begin{thebibliography}{11}

\bibitem{Ulukus-2015-review}
S. Ulukus, A. Yener, E. Erkip, O. Simeone, M. Zorzi, P. Grover,  and K. Huang, ``Energy harvesting wireless communications: a review of recent advances,"  \textit{IEEE J. Sel. Areas Commun.,} vol. 33, no. 3, pp. 360--381, Mar. 2015.

\bibitem{WSN-2011-survey}
S. Sudevalayam and P. Kulkarni, ``Energy harvesting sensor nodes: survey and implications," \textit{IEEE Commun. Surveys Tuts.}, vol. 13, no.
3, pp. 443--461, Mar. 2011.


\bibitem{Ulukus-2011-policy}
O. Ozel, K. Tutuncuoglu, J. Yang, S. Ulukus, and A. Yener,  ``Transmission with energy harvesting nodes in fading wireless
channels: Optimal policies," \textit{IEEE J. Sel. Areas Commun.}, vol. 29, no. 8, pp. 1732--1743, Sep. 2011.
\bibitem{Ulukus-2012-packet_mac}
 J. Yang and S. Ulukus, ``Optimal packet scheduling in a multiple access channel with energy harvesting transmitters," \textit{J. Commun. Netw.},  vol. 14, no. 4, pp. 140--150, Apr. 2012.
\bibitem{RuiZhang-2012-tsp}
C. K. Ho and R. Zhang, ``Optimal energy allocation for wireless communications with energy harvesting constraints," \textit{IEEE Trans. Signal Process.}, vol. 60, no. 9, pp. 4808--4818, Sep. 2012.

%%%\bibitem{BT-2014-JCN}
%%%B. T. Bacinoglu and E.U. Biyikoglu, ``Finite-horizon online transmission scheduling on an energy harvesting communication link with a discrete set of rates,¡± \textit{J. Commun. Netw.}, vol. 16, no. 3, pp. 293--300, Mar. 2014.

\bibitem{Ydong-2015-JSAC}
Y. Dong, F. Farnia and A. Ozgur, ``Near optimal energy control and approximate capacity of energy harvesting communication," \textit{IEEE J. Sel. Areas Commun.,} vol. 33, no. 3, pp. 540--557, Mar. 2015.

\bibitem{Zeng-2015-TWC}
W. Zeng, Y. R. Zheng, and R. Schober, ``Online resource allocation for energy harvesting downlink multiuser systems: precoding with modulation, coding rate, and subchannel selection," \textit{IEEE Trans. Wireless Commun.}, vol. 14, no. 10, pp. 5780--5794, Oct. 2015.

\bibitem{Badiei-2014-TIT}
M. B. Khuzani and P. Mitran, ``On online energy harvesting in multiple access communication systems," \textit{IEEE Trans. Inform. Theory}, vol. 60, no. 3, pp. 1883--1898, Mar. 2014.

\bibitem{YDong-2016-JSAC}
Y. Dong, J. Wang, B. Shim, and D.  I.  Kim, ``DEARER: A distance-and-energy-aware routing with energy reservation for energy harvesting wireless sensor networks," \textit{IEEE J. Sel. Areas Commun.,} vol. 34, no. 12, pp. 3798--3813,  Dec.  2016.

\bibitem{Sakulkar-2016-Arxiv}
P. Sakulkar and B. Krishnamachari, ``Online learning schemes for power allocation in energy harvesting communications," [Online]. Available: arXiv:1607.02552v2.

\bibitem{Dongin-TWC-2016}
K. W.  Choi and D. I. Kim, ``Stochastic optimal control for wireless powered communication networks,"
\textit{IEEE Trans. Wireless Commun.}, vol. 15, no. 1, pp. 686--698, Jan. 2016.

\bibitem{ZChen-2016-JSAC}
Z. Chen, Y. Dong, P. Fan, and K. B. Letaief, ``Optimal throughput for two-way relaying: energy harvesting and energy co-operation,"  \textit{IEEE J. Sel. Areas Commun.}, vol. 34, no. 5, pp. 1448--1462, May 2016.

\bibitem{Cui-2014-TWC}
C. Huang, R. Zhang, and S. Cui, ``Optimal power allocation for outage probability minimization in fading channels with energy harvesting constraints,"
\textit{IEEE Trans. Wireless Commun.}, vol. 13, no. 2, pp. 1074--1087, Feb. 2014.

\bibitem{Yzhao-2015-TWC}
Y. Zhao, B. Chen, and R. Zhang, ``Optimal power management for remote estimation with an energy harvesting sensor," \textit{IEEE Trans. Wireless Commun.}, vol. 14, no. 11, pp. 6471--6480, Nov. 2015.

\bibitem{Ozcelik-2016-ISIT}
A. \"{O}z\c{c}elikkale, T. McKelvey, and M. Viberg, ``Performance bounds for remote estimation with an energy harvesting sensor," in \textit{Proc. IEEE Int. Symp. Inf. Theory (ISIT)}, Barcelona, Spain, July 2016, pp. 460--464.

\bibitem{Bhat-2016-ICC}
R. V. Bhat, M. Motani, and T. J. Lim, ``Distortion minimization in energy harvesting sensor nodes with compression power constraints,"  in \textit{Proc. IEEE Int. Conf. Commun. (ICC)}, Kuala Lumpur, Malaysia, May 2016, pp. 1--6.

\bibitem{Nourian-2015-JSAC}
M. Nourian, S. Dey, and A. Ahl¨¦n, ``Distortion minimization in multi-sensor estimation with energy harvesting,"  \textit{IEEE J. Sel. Areas Commun.}, vol. 33, no. 3, pp. 524--539, Mar. 2015.

\bibitem{Knorn-2015-TSP}
S. Knorn, S. Dey, A. Ahl¨¦n, and D. E. Quevedo, ``Distortion minimization in multi-Sensor estimation using energy harvesting and energy sharing," \textit{IEEE Trans. Signal Process}, vol. 63, no. 11, pp. 2848--2863, Nov. 2015.

\bibitem{NIT-2011}
A. El Gamal and Y. H. Kim, \textit{Network Information Theory}, Cambridge, UK, Cambridge University Press, 2012.

\bibitem{Deniz-2015}
R. Gangula, D. G\"{u}nd\"{u}z, and D. Gesbert, ``Distributed compression and transmission with energy harvesting sensors,"
in \textit{Proc. IEEE Int. Symp. Inf. Theory} (\textit{ISIT}),  Hong Kong, China, Jun. 2015, pp. 1139--1143.

\bibitem{Dong-2016-ICCC}
Y. Dong, J. Wang, and B. Shim, ``Transmitting correlated sources using energy harvesting transmitters," in \textit{Proc. IEEE/CIC Int. Conf. Commun. China (ICCC)}, Chengdu, China, July 2016, pp. 1--6.

\bibitem{Cui-2007-TSP}
S. Cui, J.-J. Xiao, A. J. Goldsmith, Z.-Q. Luo and H. V. Poor, ``Estimation diversity and energy efficiency in distributed sensing," \textit{IEEE Trans. Signal Process.}, vol. 55, no. 9, pp. 4683--4695, Sep. 2007.

\bibitem{Pturmsn-2014-MDP}
M. L. Puterman, \textit{Markov decision processes: discrete stochastic dynamic programming}, New York, NY, USA, John Wiley \& Sons Inc.,  2014.

\bibitem{button-2015}
R. Ranjusha, et al. ``Fabrication and performance evaluation of button cell supercapacitors based on MnO$_2$ nanowire/carbon nanobead electrodes," \textit{RSC Advances} vol. 38, no.3, pp. 17492--17499, Mar. 2013.

%\bibitem{EH-2006}
%N. G . Stephen, ``On energy harvesting from ambient vibration,'' \textit{J.  Sound and Vibration}, vol. 25, no. 3, pp. 293--409, Feb. 2006.

\bibitem{Luo-2007}
J. J. Xiao and Z. Q. Luo, ``Multi-terminal source-channel communication over an orthogonal Multiple-access channel," \textit{IEEE Trans.  Inform. Theory}, vol. 53, no. 9, pp. 3255--3264, Sep. 2007.

\bibitem{cv_byod-2004}
S. Boyd and L. Vandenberghe, \textit{Convex optimization},  New York, NY, USA, Cambridge University Press, 2004.

\bibitem{Banach-1922}
S. Banach, ``Sur les op\'{e}rations dans les ensembles abstraits et leur application aux \'{e}quations int\'{e}grales," \textit{Fund. Math.}, no.3, pp. 133--181, Mar. 1922.

\bibitem{GWF-2006}
O. Kaya and S. Ulukus, ``Achieving the capacity region boundary of fading CDMA channels via generalized iterative waterfilling,"\textit{IEEE Trans. Wireless Commun.}, vol. 5, no. 11, pp. 3215--3223, Nov. 2006.

\bibitem{Fan-SCT}
P. Fan, \textit{Stochastic Processes: Theory and Applications}, Beijing, China, Press of Tsinghua University,  April. 2006.

\bibitem{Ulukus-2012-awgn}
O. Ozel and S. Ulukus, ``Achieving AWGN capacity under stochastic energy harvesting," \textit{IEEE Trans. on Inform. Theory}, vol. 58, no. 10, pp. 6471--6483, Oct. 2012.

\bibitem{Ulukus-unitsize}
K. Tutuncuoglu, O. Ozel, A. Yener, and S. Ulukus, ``The binary energy harvesting channel with a unit-sized battery," \textit{IEEE Trans.  Inform. Theory}, vol. 58, no. 99, pp. 1--1, Apl. 2017.



%%%\bibitem{Deniz-2013}
%%%O. Orhan, D. Gunduz, and E. Erkip, ``Delay-constrained distortion minimization for energy harvesting transmission over a fading channel,"
%%%in \textit{Proc. IEEE Int. Symp. Inf. Theory} (\textit{ISIT}'13), Istanbul, Turkey, July, 2013, pp. 1794--1798.

%%%\bibitem{meanvalue-1998}
%%%P. K. Sahoo and T. Riedel, \textit{Mean value theorems and functional equations}, World Scientific Publishing Co. Inc., River Edge, NJ, 1998.
\end{thebibliography}

}

%\begin{biography}[{\includegraphics[width=1in,height=1.25in,clip,keepaspectratio]
%{yunquan_dong.eps}}]{Yunquan Dong} (M'15) received the M.S. degree in communication and information systems from the Beijing University of Posts and Telecommunications,
%Beijing, China, in 2008, and the Ph.D. degree in communication and information engineering from Tsinghua University, Beijing, in 2014.
%    He was a BK Assistant Professor with the Department of Electrical and Computer Engineering, Seoul National University, Seoul, South Korea.
%He is currently a Professor with the School of Electronic and Information Engineering, Nanjing University of Information Science and Technology, China.
%
%His research interests include heterogeneous cellular networks and energy harvesting communication systems.
%    He was a recipient of the Best Paper Award of the IEEE ICCT in 2011, the National Scholarship for Postgraduates from China¡¯s Ministry of Education in 2013, the Outstanding Graduate Award of Beijing with honors in 2014, and the Young Star of Information Theory Award from the China¡¯s Information Theory Society in 2014.
%\end{biography}
\end{document}